\documentclass{amsart}

\usepackage{amsfonts}
\usepackage{mathrsfs}
\usepackage{amsmath,amssymb,amsthm,amsfonts}
\usepackage{cite}
\usepackage{caption}
\usepackage{pgfplots}
\usepackage{enumerate}
\usepackage{subfig}
\pgfplotsset{plot coordinates/math parser=false}
\newlength\figureheight
\newlength\figurewidth

\usepackage{booktabs, threeparttable, stackengine}

\usepackage{tikz}

\RequirePackage[OT1]{fontenc}
\RequirePackage[colorlinks,citecolor=blue,urlcolor=blue]{hyperref}

\DeclareMathOperator*{\vecc}{vec}

\newcommand{\overbar}[1]{\mkern 1.5mu\overline{\mkern-1.5mu#1\mkern-1.5mu}\mkern 1.5mu}

\numberwithin{equation}{section}
\theoremstyle{plain}
\newtheorem{theorem}{Theorem}[section]
\newtheorem{lemma}{Lemma}[section]

\newtheorem{prop}{Proposition}[section]
\theoremstyle{definition}

\theoremstyle{remark}
\newtheorem{remark}{Remark}[section]
\newtheorem{example}{Example}

\theoremstyle{remark}
\newtheorem{simt}{Simulation}
\allowdisplaybreaks

\newcommand{\floor}[1]{\lfloor #1 \rfloor}

\newcommand{\norm}[1]{\left\lVert#1\right\rVert}
\newcommand{\Log}[2][]{{\mathrm Log}_{#1}{\left(#2\right)}}
\newcommand{\emet}[2]{\langle #1,#2 \rangle}
\newcommand{\Exp}[2][]{{\mathrm Exp}_{#1}{\left(#2\right)}}

\pdfoutput=1

\pgfplotsset{compat=1.14}
\begin{document}

\title{Robust   optimization  and inference on manifolds}

\author[Lin et al.]{ Lizhen Lin, Drew Lazar, Bayan Sarpabayeva, and David B. Dunson}

\address{Department of Applied and Computational Mathematics and Statistics, The University of Notre Dame,  Notre Dame, IN. }

\email{lizhen.lin@nd.edu}


\begin{abstract}
We propose a robust and scalable procedure for general optimization and inference problems  on manifolds  leveraging the classical idea of `median-of-means' estimation. This  is motivated by ubiquitous examples and applications in modern data science in which a  statistical  learning problem can be cast as an optimization problem over manifolds. Being able to incorporate the underlying geometry for inference while addressing the need for robustness and scalability presents great challenges. We address these challenges by first proving a key lemma that characterizes some crucial properties of geometric medians  on manifolds. In turn, this allows us to prove  robustness and tighter concentration of our proposed  final estimator in a subsequent theorem.  This estimator aggregates a collection of subset estimators by taking their geometric median over the manifold. We illustrate bounds on this estimator via calculations in explicit examples.  The robustness and scalability of the procedure is illustrated in numerical examples on both simulated and real data sets.

\end{abstract}

\maketitle



\vspace*{.3in}


\newpage


\section{Introduction}

There is a  rapidly growing collection of learning problems and applications in data science that can be formalized as  optimization problems over non-Euclidean spaces, such as non-linear Riemannian manifolds. Advancement in technology and computing leads to the increasing prevalence of complex data that are in non-Euclidean forms, such as positive definite matrices (diffusion matrices) in diffusion tensor imaging~\cite{dti-ref}, shape objects in medical vision~\cite{kendall}, network data objects  \cite{paperwitqeric}, subspaces or orthonormal  frames and so on~\cite{sinicapaper}. Proper statistical inference from such data involves optimization over the underlying manifold to which the data are constrained. For example, there is a vibrant line of research based on estimation of Fr\'echet means~\cite{frechet},  which are minimizers of Fr\'echet functions on manifolds~\cite{rabibook, linclt}. In this case, both the data and parameters of interest are on manifolds. In addition, it is common to represent lower-dimensional structure in high-dimensional data as a manifold.  Learning such a manifold is a non-trivial optimization problem.  In each of the above problems, developing algorithms that are robust to data contamination and heavy tails and that scale efficiently to large datasets is crucial.


With this motivation, our main aim is to propose a robust and scalable procedure for general optimization on manifolds. We generalize the powerful `median-of-means' estimator~\cite{Nemirovski1983Problem-complex00},  to manifolds by establishing some key properties of the geometric median on manifolds with which we can prove tighter concentration bounds of our proposed estimator. The key idea is to obtain optimizers from subset data which are aggregated to form a final estimator. Our estimator can be shown to be robust to outliers and contaminations of arbitrary nature and has provable robustness. Scalability of the algorithm is automatically gained via the divide-and-conquer nature of combining subset-based estimators.

There is a related literature outside of the non-Euclidean manifold setting.  For example, \cite{minsker2015} applies the `median-of-means' procedure for robust estimation in Banach spaces.  In~\cite{median-posterior}, a robust Bayesian estimator is proposed as the geometric median of subset posteriors measures. There has been recent  theoretical and computational developments on applying the median of mean procedure in learning theory~\cite{lecu2020, lugosi2019}.  Characterizing properties of the geometric median on manifolds requires a substantially different approach, which deals with the underlying geometry.  We prove a key lemma characterizing the robustness property of geometric medians on manifolds, which allows us to show our estimator has tighter concentration bounds than subset estimators.  This is done for both the \emph{extrinsic geometric median} and the \emph{intrinsic geometric median} with the former employing an embedding of manifolds into some higher-dimensional Euclidean space and the later adopting a Riemannian structure. 
We illustrate the bounds with explicit calculations in both the extrinsic and intrinsic cases. Our procedure is demonstrated in a class of manifolds through both simulated and real data examples. Manifolds considered include the sphere, positive definite matrices and the planar shape spaces, all of which are commonly applicable in real data analyses.

The paper is organized as follows: in section 2 we  introduce the general procedure and prove a key property of the geometric median on manifolds.  Section 3 is devoted to robust estimation and optimization on manifolds.  In particular, we prove the concentration property of our final estimator in estimating the population parameter of interest and provide examples of calculations of the bounds.  In section 4 we consider an extensive simulation study and data analysis illustrating both the robustness and scalability of our procedure. The papers ends with a discussion.

\section{Geometric median and robust estimation on manifolds }


Let $Q$ be a probability distribution on some space $\mathcal X$ and $\mathcal M$ be a manifold.  We consider the problem of estimating the {\it population parameter}
\begin{equation}
\label{eqfrechet}
\mu=\arg\min_{p\in\mathcal M} L^*(p),
\end{equation}
where $L^*(p)$ is defined as
\begin{align*}
L^*(p)= \int_{\mathcal X} L(p, x)Q(dx)
\end{align*}
for some loss function $ L$.

Let $\boldsymbol x=\{x_1,\ldots, x_n\}$ where $x_1,\ldots, x_n$ are sampled from $Q$.  The parameter $\mu$ is often estimated by the {\it empirical risk estimator}
\begin{align} \label{eq:emprisk}
\hat{\mu}_n=\arg\min_{p\in\mathcal{M}}L_n(p,\boldsymbol x)=\arg\min_{p\in\mathcal{M}}\frac{1}{n} \sum_{i=1}^nL(p,x_i).
\end{align}

\begin{remark}
An important example is the {\it Fr\'echet mean} in which the risk function is 
\begin{align*}
L^*(p)=\int \rho^2(p, x)Q(dx),
\end{align*}
with $Q$ supported on a manifold $\mathcal X = \mathcal M$ and $\rho$ a metric defined on $\mathcal M$, and $\hat{\mu}_n$ corresponds to the sample Fr\'echet mean.  There is significant literature on nonparametric statistical inference on manifolds in which estimation of the Fr\'echet mean is addressed (see \cite{rabibook, linclt}). Similarly, in a regression problem with manifold-valued output, the underlying problem can be cast as an optimization problem on manifolds~\cite{linregression}.
In many other applications, we do not have $\mathcal{X} = \mathcal{M}$ with $\mathcal X$ a higher-dimensional ambient space and optimization done over a lower-dimensional manifold such as the Grassmannian~\cite{Lohit2017LearningIR,Saparbayeva2018CommunicationEP}, which has abundant applications in  manifold learning and low-rank estimation matrix problems~\cite{lowrank, BOUMAL2015}.
\end{remark}
Real data sets often contain outliers that can be errors, extreme observations or contamination of various sorts which occur when sampling from heavy tailed or mixture distributions. Thus, there is interest in robust estimation of population parameters by estimators which are stable and not unduly effected by the presence of outliers.

In this paper, we consider the classic and intuitive estimator formed by taking the geometric median of a collection of subset estimators or optimizers. Before formally introducing our procedure in the next section, we introduce the notion of the {\it geometric median on a manifold} and prove an important lemma about its properties.


For a metric space $(\mathcal{M}, \rho)$ the {\it geometric median}, $p^{*}$, of points $p_1, \ldots, p_m\in\mathcal{M}$ minimizes the sum of distances to the points, i.e.,
\begin{equation} \label{geo_med}
    p^{*}={\rm med}(p_1, \ldots, p_m)=\arg\min_{p\in\mathcal{M}}\frac{1}{m}\sum_{k=1}^m\rho(p, p_k)
    \end{equation}
assuming that $ p^{*}$ exists and is unique.
When $\mathcal M$ is a manifold, there are different ways to metrize the space. Let $J: \mathcal{M}\rightarrow \mathbb{R}^D$ be an embedding of a manifold $\mathcal{M}$ into some higher-dimensional Euclidean space $\mathbb{R}^D.$   One can define an {\it extrinsic distance} on $M$ induced from the embedding $J$ in which
\begin{align*}
\rho(p,q)=\|J(p)-J(q)\|,
\end{align*}
where $\|\cdot\|$ is the Euclidean norm on $\mathbb R^D$.
Alternatively,  one can take $\rho$ to be the {\it intrinsic distance} as the geodesic distance arising from a Riemannian structure on $\mathcal{M}$.

With the choice of $\rho$ as the extrinsic or intrinsic distance in \eqref{geo_med},  we have corresponding definitions of the \emph{extrinsic geometric median} and  the \emph{intrinsic geometric median}, respectively. Some properties of the intrinsic geometric median are studied in~\cite{FletcherVJ08} by, for example,  characterizing the uniqueness  conditions of the intrinsic sample median along with a Weizfeld algorithm for finding the median. Our theoretical results below on robustness are of a fundamentally different nature, allowing us to construct an estimator that is not only robust but also has tighter bounds around the true parameter of interest.


We prove the following lemma, which says if $\omega \in \mathcal M$  is at least a constant, $C_\alpha$, times $\epsilon$ distance away from the geometric median $p^{*}={\rm med}(p_1, \ldots, p_m),$ then $\omega$  is at least $\epsilon$ distance away from at least an $\alpha$ fraction of the points $p_1, \ldots, p_m$.  This result is illustrated in Figure \ref{fig1}.  A similar result was proved in
\cite{minsker2015} 
in the case of Banach spaces. The proof of the following, a general lemma for manifolds, requires additional machinery.

\begin{figure}[ht!]
\begin{center}
\begin{tikzpicture} [scale=1.5]

    \draw[smooth cycle,tension=.7] plot coordinates{(.45,0) (0.95,1.7) (2.2,1.7) (4.2,2.3) (4.7,0)};
    \coordinate (A) at (1,1);
    \draw (A) arc(140:40:.5) (A) arc(-140:-20:.5) (A) arc(-140:-160:.5);

\coordinate (p1) at (2.3,1.46);
\coordinate (p2) at (4.3,0.464);

\draw[-,domain=1:1.4, blue, variable=\t, samples=200] plot({\t+.9},{\t^2-.5});
\draw[-,domain=-.1:.2, blue, variable=\t, samples=200] plot({8*\t+2.7},{-4*(\t^3)+.496});
\draw [blue] (p1) edge[bend left=7,looseness=1] (p2);
\node at (4.35,0.33) {\footnotesize $p^*$};
\node at (1.9,0.35) {\footnotesize$\omega$};
\node at (2.4,1.6) {\footnotesize$p_j$};
\draw [->](3,.2) -- (2.9,.45);
\node at (3,.1) {\footnotesize$\rho(\omega,p^*)\geq C_{\alpha}\epsilon$};
\draw[fill,blue] (4.3,0.464) circle [radius=1pt];
\draw[fill,green] (1.9,0.5) circle [radius=1pt];
\draw[fill,red] (2.3,1.46) circle [radius=1pt];
\draw [->](1.8,1.35) -- (2.15,1.15);
\node at (1.27,1.37) {\footnotesize$\rho(w,p_j)\geq \epsilon$};
\node at (3.35,1.67) {\small $\mathcal{M}$};
\end{tikzpicture}
  \caption{Geometric Illustration of Lemma \ref{leml1} on Manifold $\,\mathcal{M}$}
  \label{fig1}

  \end{center}
\end{figure}
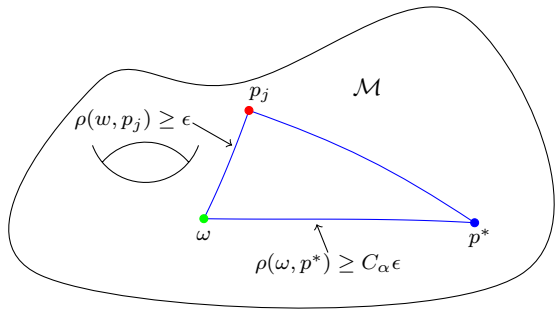

\begin{lemma}
\label{leml1}
Let $p_1, \ldots, p_m \in \mathcal M$, $p^{*}={\rm med}(p_1, \ldots, p_m)$ as in \eqref{geo_med}. Then (a) and (b) below hold.
\begin{itemize}
\item[(a)]  Let $\rho$ be the extrinsic distance for some embedding $J: \mathcal{M}\rightarrow \tilde{\mathcal{M}}\subset \mathbb R^D$. Let $\omega\in\mathcal{M}$, $\psi$ be angle between $J(\omega)-J(p^*)$ and the tangent space $T_{J(p^*)}\tilde{\mathcal{M}}$ and let
    \[
    C_{\alpha}=\dfrac{1-\alpha}{\sqrt{1-2\alpha}\cos\psi-\alpha\sin\psi}
    \]
     where  $\alpha\in(0, \cot\psi\tan\frac{\psi}{2}).$ If
 $\rho(\omega, p^*)\geq C_{\alpha}\epsilon,$ then there exists an $\alpha$ portion of elements of $p_1, \ldots, p_m$ which are at least $\epsilon$ distance away from $\omega$. That is, there exists an index set $T\subset \{1,\ldots, m\}$  with $|T|\geq \alpha m$, and $\rho(p_j, \omega)\geq \epsilon$ for any $j\in T$.

 \item[(b)] Let  $\rho$  be an intrinsic distance  on $M$ with respect to some Riemannian structure. Let $\omega\in\mathcal{M}$, the $\log$ map $\log_{p^*}$\,be $K$-Lipschitz continuous from $B(\omega, \epsilon)$ to $T_{p^*}\mathcal{M}$ and let
     \[
     C_{\alpha}= K(1-\alpha)\sqrt{\frac{1}{1-2\alpha}}
     \]
     where $\alpha\in(0, 1/2)$. If
 $\rho(\omega, p^*)\geq C_{\alpha}\epsilon,$
 then there exists an $\alpha$ portion of elements of $p_1, \ldots, p_m$ which are at least $\epsilon$ distance away from $\omega$.
 \end{itemize}
\end{lemma}

\begin{proof}
(a)
Let $L(J(p))=\sum_{j=1}^m\rho(p, p_j)=\sum_{j=1}^m \|J(p)-J(p_j)\|$ for $J(p) \in \tilde{\mathcal{M}}$. Let $\gamma(t)$ be a  curve from $J(p^*)$ to $J(\omega)$ on $\tilde{\mathcal{M}}$, where $\gamma(0)=J(p^*)$,
$\gamma(1)=J(\omega)$, and $\gamma'(0)=v$. The directional derivative of $L$ at $J(p^*)$ evaluated at $v$ is given by
\begin{equation} \label{Datpstar}
\begin{aligned}
dL_{J(p^*)}(v)&=\lim_{t \rightarrow 0^+}\dfrac{L\left(\gamma(t)\right)- L\left(\gamma(0)\right)}{t} \\
&= \lim_{t\rightarrow 0^+}\dfrac{L\left(\gamma(t)\right)-L(J(p^*))}{t}\geq 0
\end{aligned}
\end{equation}
with the above inequality holding as $J(p^*)$ minimizes $L$ for $p\in\mathcal{M}.$
Let
\begin{equation*}
 \gamma(t)=\mathcal{P}_{\tilde{\mathcal{M}}} \Big(J(p^*)+t\big(J(\omega)-J(p^*)\big)\Big),
\end{equation*}
 where $\mathcal{P}$ is the projection of $\mathbb{R}^D$ onto $\tilde{\mathcal M}$, that is,
\begin{equation*}
    \mathcal{P}(x)=\arg\min_{y\in\tilde{\mathcal{M}}}\rho(y, x).
\end{equation*}
 We assume the projection map $\mathcal{P}$ is differentiable at $t=0$.  Denote $ \mathcal J$ as the Jacobian matrix of the projection map $\mathcal P$ at $J(p^*)$. Then one has
\begin{align*}
v=\gamma'(0)=\mathcal J \big(J(\omega)-J(p^*)\big),
\end{align*}
which will be needed in determining the constant $C_{\alpha}.$
One can  see that
\begin{align*}
L\left(\gamma(t)\right)-L(J(p^*))=\sum_{j=1}^m \left(\|\gamma(t)-J(p_j)\|-\|\gamma(0)-J(p_j)\| \right).
\end{align*}
Let
\begin{align*}
A_j=\dfrac{\|\gamma(t)-J(p_j)\|-\|\gamma(0)-J(p_j)\| }{t} \, \, \text{for } j=1, \ldots, m.
\end{align*}
Then
\begin{align*}
A_j=\dfrac{\|\gamma(t)-J(p_j)\|^2-\|\gamma(0)-J(p_j)\|^2 }{t\left(\|\gamma(t)-J(p_j)\|+\|\gamma(0)-J(p_j)\|\right)} \, \, \text{for } j=1, \ldots, m.
\end{align*}
One has
\begin{align}
\label{eqdeno}
\lim_{t\rightarrow 0^+}\left(\|\gamma(t)-J(p_j)\|+\|\gamma(0)-J(p_j)\|\right)=2\|\gamma(0)-J(p_j)\|.
\end{align}
Also,
\begin{align*}
\|\gamma(t)-J(p_j)\|^2&=\langle \gamma(t)-J(p_j), \gamma(t)-J(p_j)\rangle\\
&=\langle\gamma(t),\gamma(t)\rangle -2\langle\gamma(t), J(p_j)\rangle+\langle J(p_j), J(p_j)\rangle,
\end{align*}
and
\begin{align*}
\|\gamma(0)-J(p_j)\|^2=\langle\gamma(0), \gamma(0)\rangle-2\langle\gamma(0), J(p_j)\rangle+ \langle J(p_j), J(p_j)\rangle.
\end{align*}
Then
\begin{align*}
&\|\gamma(t)-J(p_j)\|^2-\|\gamma(0)-J(p_j)\|^2\\
 \quad  \quad &=\langle\gamma(t), \gamma(t)\rangle -\langle\gamma(0), \gamma(0)\rangle-2\langle\gamma(t)-\gamma(0), J(p_j)\rangle \\
 &=\left ( \langle\gamma(t), \gamma(t)\rangle-\langle\gamma(0), \gamma(t)\rangle\right)+\left(\langle\gamma(0), \gamma(t)\rangle-\langle\gamma(0), \gamma(0)\rangle\right)\\
 &\qquad-2\langle\gamma(t)-\gamma(0), J(p_j)\rangle\\
 &=\langle\gamma(t)-\gamma(0), \gamma(t) \rangle+\langle \gamma(0), \gamma(t)-\gamma(0)\rangle-2\langle\gamma(t)-\gamma(0), J(p_j)\rangle\\
 &=\langle\gamma(t)-\gamma(0), \gamma(t)+\gamma(0)-2J(p_j)\rangle.
\end{align*}
Therefore,
\begin{align*}
\lim_{t\rightarrow 0^+}\frac{\|\gamma(t)-J(p_j)\|^2-\|\gamma(0)-J(p_j)\|^2}{t}&=\lim_{t\rightarrow 0^+}\left\langle\frac{\gamma(t)-\gamma(0)}{t}, \gamma(t)+\gamma(0)-2J(p_j)\right\rangle\\
&=\langle\gamma'(0), \gamma(0)+\gamma(0)-2J(p_j)\rangle\\
&=2\langle\gamma'(0),\gamma(0)-J(p_j)\rangle = 2\langle\gamma'(0), J(p^*)-J(p_j)\rangle.
\end{align*}
Thus, by \eqref{eqdeno} and the above equation, if $J(p_j)\neq J(p^*)$, one has
\begin{align*}
\lim_{t\rightarrow 0^+}A_j=\dfrac{\langle\gamma'(0), J(p^*)-J(p_j)\rangle}{\|J(p^*)-J(p_j)\|}.
\end{align*}
Otherwise, if $J(p_j)=J(p^*),$ then
\begin{align*}
\lim_{t\rightarrow 0^+}A_j=\lim_{t\rightarrow 0^+}\dfrac{\|\gamma(t)-J(p_j)\|}{t}=\|\gamma'(0)\|.
\end{align*}
Therefore,
\begin{align*}
dL_{J(p^*)}(v) &=\sum_{j=1}^m\lim_{t\rightarrow 0^+}A_j\\
&=\sum_{j: p_j\neq p^*}\frac{\langle\gamma'(0), J(p^*)-J(p_j)\rangle}{\|J(p^*)-J(p_j)\|}+\|\gamma'(0)\|\sum_{j=1}^m I(p_j=p^*),
\end{align*}
where $I(\cdot)$ is the indicator function.
The above implies
\begin{align}
\label{23}
\frac{dL_{p^*}(v)}{\|\gamma'(0)\|}&=\sum_{j=1}^m\lim_{t\rightarrow 0^+}\frac{A_j}{\|\gamma'(0)\|}\\
&=\sum_{j: p_j\neq p^*}\frac{\langle\gamma'(0), J(p^*)-J(p_j)\rangle}{\|\gamma'(0)\|\|J(p^*)-J(p_j)\|}+\sum_{j=1}^m I(p_j=p^*).
\end{align}
The Jacobian matrix of the projection map $\mathcal P$ at $J(p^*)$, $ \mathcal J$, is the orthogonal projection of $T_{J(p^*)}\mathbb{R}^D \equiv \mathbb{R}^D $ to $T_{J(p^*)}\tilde{\mathcal{M}}$.
That is, for $a \in T_{J(p^*)}\mathbb{R}^D $,
  $\mathcal J(a)=a_1,$
where $a = a_1 + a_2$ is the unique orthogonal decomposition of $a$ with $a_1 \in T_{J(p^*)}\tilde{\mathcal{M}} $.
Now assume that there does \textit{not} exist an $\alpha$ portion of elements of $p_1,\ldots, p_m$ which are at least $\epsilon$ distance away from $\omega$, that is, without loss of generality,
\[
\|J(p_j)-J(\omega)\|\leq \epsilon \, \, \text{for } j=1,\ldots, \floor{(1-\alpha)m}+1.
\] Let us denote by $\angle \big(J(\omega)-J(p*), J(p_j)-J(p^*)\big)$ the angle between the vectors $J(\omega)-J(p^*)$ and $J(p_j)-J(p^*).$
Then for $j=1,\ldots, \floor{(1-\alpha)m}+1,$
\begin{equation*}
  \sin\Big(\angle\big(J(\omega)-J(p^*), J(p_j)-J(p^*)\big)\Big)<\frac{1}{C_{\alpha}}
\end{equation*}
and so
\begin{equation*}
  \cos\Big(\angle\big(J(\omega)-J(p^*), J(p_j)-J(p^*)\big)\Big)>\sqrt{1-\frac{1}{C_{\alpha}^2}}.
\end{equation*}
Notice that
\begin{multline*}
  \angle\Big(\mathcal{J}\big(J(\omega)\big)-J(p^*), J(\omega)-J(p^*)\Big)+\angle\big(J(\omega)-J(p^*), J(p_j)-J(p^*)\big)\\
  = \psi+\angle\big( J(\omega)-J(p^*), J(p_j)-J(p^*)\big) \geq \ \angle\Big(\mathcal{J}\big(J(\omega)\big)-J(p^*), J(p_j)-J(p^*)\Big).
\end{multline*}
Therefore,
\begin{multline*}
  \cos\bigg(\angle\Big(\mathcal{J}\big(J(\omega)\big)-J(p^*), J(p_j)-J(p^*)\Big)\bigg)\geq \cos\Big(\psi+\angle\big(J(\omega)-J(p^*), J(p_j)-J(p^*)\big)\Big)\\
  >\sqrt{1-\frac{1}{C_{\alpha}^2}}\cos\psi-\frac{1}{C_{\alpha}}\sin\psi.
\end{multline*}
We have
\begin{align*}
\begin{aligned}
  \frac{\langle\gamma'(0), p_j-J(p^*)\rangle}{\|\gamma'(0)\|\|p_j-J(p^*)\|}&=\cos\bigg(\angle\Big(\mathcal{J}\big(J(\omega)\big)-J(p^*), J(p_j)-J(p^*)\Big)\bigg)\\ &>\sqrt{1-\frac{1}{C_{\alpha}^2}}\cos\psi-\frac{1}{C_{\alpha}}\sin\psi.
\end{aligned}
\end{align*}

Then for any $\alpha\in\Big(0, \cot\psi\tan\frac{\psi}{2}\Big)$ from \eqref{23}
\begin{equation*}
  \dfrac{dL_{J(p^*)}(v)}{\|\gamma'(0)\|}<-(1-\alpha)m\bigg(\sqrt{1-\frac{1}{C_{\alpha}^2}}\cos\psi-\frac{1}{C_{\alpha}}\sin\psi\bigg)+\alpha m\leq0,
\end{equation*}
when
\begin{equation*}
 C_{\alpha}\geq\frac{1-\alpha}{\sqrt{1-2\alpha}\cos\psi-\alpha\sin\psi}
\end{equation*}
which is a contradiction with \eqref{Datpstar}.

(b) The intrinsic median requires a different proof. Let $L(p)=\sum_{j=1}^m\rho(p, p_j)$ where $\rho$ is the intrinsic distance; we use the Riemannian exponential map $\exp_{p*}:T_{p*}\mathcal M\rightarrow\mathcal M.$ Let $v=\log_{p^{*}}\omega\in T_{p*}\mathcal M $ and consider the geodesic curve $\gamma(t)=\exp_{p*}(tv).$  Then
\begin{align} \label{eq:intd}
dL_{p^*}(v)=\lim_{t \rightarrow 0}\frac{L(\gamma(t))-L(\gamma(0))}{t}=\lim_{t \rightarrow 0}\dfrac{L(\gamma(t))-L(p^*)}{t}\geq 0.
\end{align}
Denote
\begin{equation*}
A=\lim_{t \rightarrow 0+} \sum_{j=1}^m \left( \frac{\sqrt{\big\langle\gamma_{js}(s, t), \gamma_{js}(s, t)\big\rangle}-\sqrt{\big\langle\gamma_{js}(s, 0), \gamma_{js}(s, 0)\big\rangle}}{t} \right),
\end{equation*}
where $\gamma_j(s, t)=\exp_{\gamma(t)}(s \log_{\gamma(t)}p_j)=\exp_{\gamma(t)}(sv_j(t))$ is the geodesic curve connecting $\gamma(t)$ with $p_j,$ then $\gamma_{js}(s, t)=\frac{\partial\gamma_j(s, t)}{\partial s}.$
Set
\begin{align*}
A_j=\frac{\sqrt{\big\langle\gamma_{js}(s, t), \gamma_{js}(s, t)\big\rangle}-\sqrt{\big\langle\gamma_{js}(s, 0), \gamma_{js}(s, 0)\big\rangle}}{t},\, \, \text{for } j=1, \ldots, m.
\end{align*}
 Then
\begin{align*}
A_j=\frac{1}{t}\frac{\langle\gamma_{js}(s, t), \gamma_{js}(s, t)\rangle-\langle\gamma_{js}(s, 0), \gamma_{js}(s, 0)\rangle}{\sqrt{\langle\gamma_{js}(s, t), \gamma_{js}(s, t)\rangle} +\sqrt{\langle\gamma_{js}(s, 0), \gamma_{js}(s, 0)\rangle}}, \, \, \text{for } j=1, \ldots, m.
\end{align*}
We see that
\begin{align*}
\lim_{t\rightarrow 0^+}\left(\sqrt{\langle\gamma_{js}(s, t), \gamma_{js}(s, t)\rangle}+\sqrt{\langle\gamma_{js}(s, 0), \gamma_{js}(s, 0)\rangle}\right)=2\sqrt{\langle\gamma_{js}(s, 0), \gamma_{js}(s, 0)\rangle}.
\end{align*}
On the other hand,
\begin{multline*}
\lim_{t\rightarrow 0^+}\frac{\langle\gamma_{js}(s, t), \gamma_{js}(s, t)\rangle-\langle\gamma_{js}(s, 0), \gamma_{js}(s, 0)\rangle}{t} =2\Big\langle\frac{D}{dt}\gamma_{js}(s, 0), \gamma_{js}(s, 0)\Big\rangle\\
=2\Big\langle\frac{D}{ds}\gamma_{jt}(s, 0), \gamma_{js}(s, 0)\Big\rangle =2\frac{d}{ds}\big\langle\gamma_{jt}(s, 0), \gamma_{js}(s, 0)\big\rangle.
\end{multline*}
Thus  if $p_j\neq p^*$, one has
\begin{align*}
\lim_{t\rightarrow 0^+}A_j=\frac{\frac{d}{ds}\langle\gamma_{jt}(s, 0), \gamma_{js}(s, 0)\rangle}{\sqrt{\langle\gamma_{js}(s, 0), \gamma_{js}(s, 0)\rangle}}.
\end{align*}
Otherwise, if $p_j= p^*$, then
\begin{align*}
\lim_{t\rightarrow 0^+}A_j= \lim_{t\rightarrow 0^+} \frac{\sqrt{\langle-t\gamma'((1-s)t), -t\gamma'((1-s)t)}}{t}= \lim_{t\rightarrow 0^+} \frac{t\|v\|}{t}=\|v\|.
\end{align*}
Therefore,
\begin{align*}
dL_{p*}(v)&=\sum_{j=1}^m \int_0^1\lim_{t\rightarrow 0^+}  A_jds\\
&=\sum_{j: p_j\neq p^*} \int_0^1\dfrac{\frac{d}{ds}\langle\gamma_{jt}(s, 0), \gamma_{js}(s, 0)\rangle}{\sqrt{\langle\gamma_{js}(s, 0), \gamma_{js}(s, 0)\rangle}}ds+\|v\|\sum_{j=1}^m I(p_j=p^*)\\
&=\sum_{j: p_j\neq p^*} \dfrac{\langle\gamma_{jt}(1, 0), \gamma_{js}(1, 0)\rangle}{\|v_j\|}+\|v\|\sum_{j=1}^m I(p_j=p^*)\\
&=\sum_{j: p_j\neq p^*} \dfrac{\langle(d\exp_{p^*})_{v_j}\big(1\cdot v_j'(0)\big), (d\exp_{p^*})_{v_j}v_j\rangle}{\|v_j\|}+\|v\|\sum_{j=1}^m I(p_j=p^*)\\
&=\sum_{j: p_j\neq p^*} \dfrac{\langle v_j'(0), v_j\rangle}{\|v_j\|}+\|v\|\sum_{j=1}^m I(p_j=p^*)\\
&=-\sum_{j: p_j\neq p^*} \dfrac{\langle v, v_j\rangle}{\|v_j\|}+\|v\|\sum_{j=1}^m I(p_j=p^*),
\end{align*}
where $I(\cdot)$ is the indicator function.
Then one has,
\begin{align*}
  \dfrac{dL_{p^*} (v) }{\|v\|}&=-\sum_{j: p_j\neq p^*}  \dfrac{\langle v, v_j\rangle}{\|v\|\|v_j\|}+\sum_{j=1}^m I(p_j=p^*)\\
  &=-\sum_{j: p_j\neq p^*}\cos(\widehat{v, v_j})+\sum_{j=1}^m I(p_j=p^*).
\end{align*}
From the condition that $\log_{p^*}$ is $K$-Lipschitz continuous from $B(\omega, r)$ to $T_{p^*}\mathcal{M}$,
\begin{equation*}
  \|v_j-v\|\leq K d_g(\exp_{p^*}v_j, \exp_{p^*}v).
\end{equation*}
Then this yields
\begin{equation*}
   \dfrac{dL_{p^*}(v)}{\|v\|}<-(1-\alpha)m\sqrt{1-\frac{K^2}{C_{\alpha}^2}}+\alpha m\leq 0,
\end{equation*}
whenever $C_{\alpha}\geq K(1-\alpha)\sqrt{\frac{1}{1-2\alpha}},$ which leads to a contradiction with \eqref{eq:intd}.
\end{proof}

There are many known Riemannian manifolds with $K$-Lipschitz continuous $\log$ maps as required in part (b) of the above lemma. Below we provide a few examples including the sphere, the planar shape space and the space of positive definite matrices,  which are commonly encountered manifolds in the statistics and medical imaging literature.

\begin{prop}
Let $S^d=\{p\in\mathbb{R}^{d+1}:\|p\|=1\}$ which is the $d$-dimensional sphere.  The inverse exponential map, $\log_p$,~on $S^d$ is 2-Lipschitz continuous from $B(p, \pi/2)$ to $T_pS^d$ for all $p \in S^d$.
\end{prop}

\begin{proof} The tangent space at $p$ is given as
\begin{equation*}
T_pS^d=\{v\in\mathbb R^{d+1}: v^{T}p=0\}.
\end{equation*}
Then for $q \in  S^d$ the inverse exponential map can be expressed as
\begin{equation*}
\log_p(q)=\frac{\arccos(p^Tq)}{\sqrt{1-(p^Tq)^2}}\big(q-(p^Tq)p\big).
\end{equation*}
Hence, the distance between $\log_pq_1$ and $\log_pq_2$ is equal to
\begin{equation*}
\|\log_pq_1-\log_pq_2\|=
\sqrt{\arccos(p^Tq_1)^2+\arccos(p^Tq_2)^2-2\arccos(p^Tq_1)\arccos(p^Tq_2)\cos\varphi}
\end{equation*}
where $\varphi$ is the angle between $\log_pq_1$ and $\log_pq_2$. The geodesic distance between $q_1$ and $q_2$ is then given by
\begin{equation*}
d_g(q_1, q_2)=\arccos(q_1^Tq_2).
\end{equation*}
One can easily obtain that
\begin{equation*}
q_1^Tq_2=(p^Tq_1)(p^Tq_2)+\sqrt{1-(p^Tq_1)^2}\sqrt{1-(p^Tq_2)^2}\cos\varphi.
\end{equation*}
Then one can check directly that
\begin{equation*}
\|\log_pq_1-\log_pq_2\|\leq2d_g(q_1, q_2).
\end{equation*}
\end{proof}

The following proposition shows that the $\log$~map in  similarity-shape spaces \cite{kendall} also satisfies the $K-$ Lipschitz condition.
\begin{prop}

The similarity or planar shape space is given as
\begin{equation} \label{eq:planarid}
\Sigma_2^k=S^{2k-3}/S^1.
\end{equation}
 The inverse exponential map, $\log_p$,~on $\Sigma_2^k$ is 2-Lipschitz continuous from $B(p, \pi/4)$ to $T_p\Sigma_2^k$ for all $p \in \Sigma_2^k$.
 \end{prop}

   \begin{proof}
$\Sigma_2^k$ is the quotient of the sphere $S^{2k-3}$ under the following group of transformations
  \begin{equation*}
 G=\left\{
 \begin{pmatrix}
 A  & \ldots & 0 \\
    & \ddots &  \\
 0  & \ldots & A
 \end{pmatrix}\in {\rm M}(2k),
 \quad A\in {\rm SO} (2)
 \right\}\simeq S^1.
 \end{equation*}
 For any $B\in G$, we have that $B=\cos t I+ \sin t \dot{I},$ where
 \begin{equation*}
 I=
 \begin{pmatrix}
 1 & 0 & \ldots  & 0 & 0  \\
 0 & 1 & \ldots  & 0 & 0  \\
    &    & \ddots &    &     \\
 0 & 0 & \ldots  & 1 & 0  \\
 0 & 0 & \ldots  & 0 & 1
 \end{pmatrix},\quad \quad
 \dot{I}=
 \begin{pmatrix}
 0  & 1 & \ldots  & 0  & 0  \\
 -1 & 0 & \ldots  & 0  & 0  \\
     &    & \ddots &     &     \\
 0  & 0 & \ldots  & 0  & 1  \\
 0  & 0 & \ldots  & -1 & 0
 \end{pmatrix}.
 \end{equation*}
 For each $p\in\Sigma_2^k$ we define the tangent space
 \begin{equation*}
 T_p\Sigma_2^k=\{v\in\mathbb R^{2k-2}: v^Tp=0, (Ip)^Tv=0\}.
 \end{equation*}
 The inverse exponential map can be expressed as
\begin{equation*}
\log_p(q)=\frac{\arccos(p^Tq)}{\sqrt{1-(p^Tq)^2}}\big(q-(p^Tq)p\big).
\end{equation*}
Hence, the distance between $\log_pq_1$ and $\log_pq_2$ is equal to
\begin{equation*}
\|\log_pq_1-\log_pq_2\|=
\sqrt{\arccos(p^Tq_1)^2+\arccos(p^Tq_2)^2-2\arccos(p^Tq_1)\arccos(p^Tq_2)\cos\varphi},
\end{equation*}
where $\varphi$ is an angle between $\log_pq_1$ and $\log_pq_2$.  The geodesic distance between $q_1$ and $q_2$ is then given by
\begin{align*}
d_g(q_1, q_2)&=\inf_{t\in(-\pi, \pi]}\arccos(q_1^T(\cos t I+\sin t \dot{I})q_2)\\
&=\arccos\sup_{t\in(-\pi, \pi]}(\cos t q_1^Tq_2+\sin t q_1^T\dot{I}q_2)\\
&=\arccos\sqrt{(q_1^Tq_2)^2+(q_1^T\dot{I}q_2)^2}.
\end{align*}
One can easily obtain that
\begin{align*}
(q_1^Tq_2)^2+(q_1^T\dot{I}q_2)^2&=\big((p^Tq_1)(p^Tq_2)+\sqrt{1-(p^Tq_1)^2}\sqrt{1-(p^Tq_2)^2}\cos\varphi\big)^2\\
&\quad+(1-(p^Tq_1)^2)(1-(p^Tq_2)^2)(\cos\psi)^2
\end{align*}
where $\psi$ is  the angle between $\log_pq_1$ and $\dot{I}\log_pq_2$. Note that
\begin{equation*}
\pi/2-\varphi\leq\psi\leq\pi/2+\varphi.
\end{equation*}
Thus $\cos\psi\geq\cos(\pi/2-\varphi)=\sin\varphi,$ and

\begin{align*}
d_g(q_1, q_2)&\geq 2\arccos\bigg(\big((p^Tq_1)(p^Tq_2)+\sqrt{1-(p^Tq_1)^2}\sqrt{1-(p^Tq_2)^2}\cos\varphi\big)^2\\
&\quad + (1-(p^Tq_1)^2)(1-(p^Tq_2)^2)(\sin\varphi)^2\bigg)^{1/2}.
\end{align*}
Then it can be verified directly that
\begin{align*}
\|\log_pq_1-\log_pq_2\| &\leq 2\arccos\bigg(\big((p^Tq_1)(p^Tq_2)+\sqrt{1-(p^Tq_1)^2}\sqrt{1-(p^Tq_2)^2}\cos\varphi\big)^2\\
&\quad+ (1-(p^Tq_1)^2)(1-(p^Tq_2)^2)(\sin\varphi)^2\bigg)^{1/2}.
\end{align*}
Thus $\|\log_pq_1-\log_pq_2\|\leq 2 d_g(q_1, q_2).$
\end{proof}

\begin{prop}\label{prop1ps}
The manifold of positive definite $n$ by $n$ matrices ${\rm PD}(n)$
has a $1$-Lipchitz continuous inverse exponential map at any $p \in {\rm PD}(n)$.
\end{prop}

\begin{proof}

We consider the Killing metric~\cite{ zbMATH05} in the manifold of invertible $n$ by $n$ matrices ${\rm GL}(n)$
\begin{equation*}
    ds^2(a)={\rm tr}(a^{-1}da)^2.
\end{equation*}
In other words, in the Lie algebra $\mathfrak{gl}(n)=T_I{\rm GL}(n)={\rm M}(n)$, we have the symmetric inner product
\begin{equation*}
    \langle A, B\rangle_{I}={\rm tr}(AB), \qquad A, B\in \mathfrak{gl}(n),
\end{equation*}
which generates the bilaterally invariant metric in the group ${\rm GL}(n)$. That is, for any $A, B\in T_g{\rm GL}(n)$ and $a \in {\rm GL}(n)$,
\begin{align*}
    \langle A, B\rangle_a=\langle a^{-1}A, a^{-1}B\rangle_I=\langle Aa^{-1}, Ba^{-1}\rangle_I={\rm tr}(a^{-1}Aa^{-1}B).
\end{align*}
Since vectors $p^{-1}A, p^{-1}B$ do not always belong to the tangent space $T_I{\rm PD}(n),$ we instead take vectors $p^{-1/2}Ap^{-1/2}$ and 
\begin{align*}
    \langle A, B\rangle_p&=\langle p^{-1/2}A, p^{-1/2}B\rangle_{p^{1/2}}\\
    &=\langle p^{-1/2}Ap^{-1/2},  p^{-1/2}Bp^{-1/2}\rangle_I\\
    &={\rm tr}(p^{-1/2}Ap^{-1}Bp^{-1/2})={\rm tr}(p^{-1}Ap^{-1}B),
\end{align*}
where $A, B\in T_p{\rm PD}(n).$ Hence we have the metric in ${\rm PD}(n)$ induced from the Killing metric in ${\rm GL}(n).$  This metric is usually known as the Fisher-Rao metric.

For this metric we have the following exponential and logarithm mappings
\begin{equation*}
\begin{split}
    \exp_pA & = p^{1/2}\exp\big(p^{-1/2}Ap^{-1/2}\big)p^{1/2},\\
    \log_pq & = p^{1/2}\log\big(p^{-1/2}qp^{-1/2}\big)p^{1/2},
\end{split}
\end{equation*}
where
\begin{align*}
\begin{split}
        \exp Y &= I+\frac{Y}{1!}+\frac{Y^2}{2!}+\ldots+\frac{Y^n}{n!}+\ldots,\\
        \log x &= (x-I)-\frac{(x-I)^2}{2}+\ldots+(-1)^{n-1}\frac{(x-I)^n}{n}+\ldots
\end{split}
\end{align*}
for any $A, Y\in {\rm Sym}(n)$ and $p, q, x\in{\rm PD}(n).$

Let $a,q_1,q_2 \in{\rm PD}(n)$. Then we have
\begin{align*}
  \|\log_a q_1-\log_a q_2\|_a & = \|\log (a^{-1/2} q_1 a^{-1/2}) -\log (a^{-1/2}q_2 a^{-1/2}) \|_I \\
                              & \le d_g(a^{-1/2} q_1 a^{-1/2},a^{-1/2}q_2 a^{-1/2}) = d_g( q_1,q_2 )
\end{align*}
where the inequality follows from the exponential metric increasing property of the Fisher-Rao metric as in \cite{bhatia2003exponential}.

\end{proof}

\section{Robust optimization on manifolds: concentration properties}

In this section, we introduce our proposed estimator, which aggregates a collection of subset optimizers of the empirical risk function.  We first divide the data set $x_1,\ldots, x_n$ into $m$ subsets $U_1,\ldots, U_m$ each of roughly size $\lfloor n/m \rfloor$. Let $\mu_{1},\ldots, \mu_{m}$ be the optimizers of the empirical risk function from each subset, $U_1,\ldots, U_m$, respectively. That is,
\begin{align} \label{eq:submeans}
\mu_{j}=\arg\min_{p\in\mathcal{M}}L_{|U_j|}(p,  U_j) \text{ for } j=1,\ldots,m
\end{align}
as in \eqref{eq:emprisk}.
Our estimator $\mu^*$ is the \emph{geometric median of the subset optimizers}, that is,
\begin{align} \label{eq:medianmeans}
\mu^*=\arg\min_{p\in M} \sum_{j=1}^m \rho(p, \mu_{j}).
\end{align}
We will show that $\mu^*$ has desired robustness properties in estimating the population parameter $\mu$.

In \cite{minsker2015} it is proven that the geometric median of a collection of weakly concentrated estimators admits a tighter deviation bound in a  Hilbert space. With the help of the Lemma 1, we generalise this result to manifolds in the following theorem.

\begin{theorem} \label{thext}
Let $\mu_{1}, \ldots, \mu_{m}$ be a collection of independent estimators of the  parameter $\mu,$
and let geometric median $\mu^{*}={\rm med}(\mu_{1}, \ldots, \mu_{m})$.
\begin{itemize}
\item[(a)]  Let $\rho$ be the extrinsic distance on $\mathcal M$ for some embedding $J: \mathcal M\rightarrow \tilde{\mathcal M}\subset \mathbb
    R^D$. Assume for any  $\omega\in \mathcal M$ the angle between   $J(\omega)-J(\mu^{*})$ and the tangent space $T_{J(\mu^{*})}\tilde{\mathcal M}$ is no bigger than $\bar{\psi}.$  For any  $\alpha\in(0, \cot\bar{\psi}\tan\frac{\bar{\psi}}{2})$ set
    \begin{equation*}
 \overbar{C}_{\alpha}=\frac{1-\alpha}{\sqrt{1-2\alpha}\cos\bar{\psi}-\alpha\sin\bar{\psi}}.
\end{equation*}
\item[(b)] Let $\rho$  be an intrinsic distance on $\mathcal M$ with respect to some Riemannian structure. Assume $\log_{\mu^{*}}$~is $K$-Lipschitz continuous from $B(\mu^{*}, \epsilon)$ to $T_{\mu^{*}}\mathcal M$. For any $\alpha\in(0, \frac{1}{2})$ set
    \begin{equation*}
\overbar{C}_{\alpha}= K(1-\alpha)\sqrt{\frac{1}{1-2\alpha}}.
\end{equation*}
\end{itemize}
Under (a) or (b), if
\begin{align} \label{eq:etainq}
P(\rho(\mu_{j}, \mu)> \epsilon)\leq \eta \text{ for } i=1,\ldots,n
\end{align}
where $\eta<\alpha$
then \begin{align} \label{eq:probgm}
P(\rho(\mu^*, \mu)>  \overbar{C}_{\alpha}\epsilon)\leq \exp(-m\phi(\alpha,\eta )),
\end{align}
where
\begin{equation*}
\phi(\alpha, \eta)=(1-\alpha)\log\frac{1-\alpha}{1-\eta}+\alpha\log\frac{\alpha}{\eta}.
\end{equation*}
\end{theorem}

\begin{proof}
Let $\psi$ be the angle between $J(\mu)-J(\mu^{*})$ and the tangent space $T_{\mu^{*}}\tilde{\mathcal M}.$ Since $\psi<\bar{\psi}$ we have $C_{\alpha}\leq\bar{C_{\alpha}}$ and $\cot\bar{\psi}\tan\frac{\bar{\psi}}{2}\leq\cot\psi\tan\frac{\psi}{2}$ where
\begin{equation*}
C_{\alpha}=\frac{1-\alpha}{\sqrt{1-2\alpha}\cos\psi-\alpha\sin\psi}.
\end{equation*}
Thus, when the event $\{\rho(\mu^{*}, \mu)>\overbar{C}_{\alpha}\epsilon\}$ occurs, the event $\{\rho(\mu^{*}, \mu)>C_{\alpha}\epsilon\}$ occurs. Then, by Lemma 1, when the event $\{\rho(\mu^{*}, \mu)>C_{\alpha}\epsilon\}$ occurs, there exists an $\alpha$ portion of elements of $\mu_{1}, \ldots, \mu_{m}$ which are at least $\epsilon$ distance away from $\mu.$ Therefore,
\begin{equation} \label{eq:gmdisbn}
P(\rho(\mu^{*}, \mu)>\bar{C_{\alpha}}\epsilon)\leq P(\rho(\mu^{*}, \mu)>C_{\alpha}\epsilon)\leq P\bigg(\sum_{j=1}^m I_{\rho(\mu_{j}, \mu)>\epsilon}>\alpha m\bigg).
\end{equation}
Let
$A=|\{j=1,...,m: \rho(\mu_{j}, \mu)>\epsilon\}|$
and let $B$ be a random variable with a binomial distribution, $B \sim b(m, \eta).$ Then with \eqref{eq:etainq} and by Lemma 23 in \cite{lerasle2011} there exists a coupling $C=(\tilde A, \tilde B)$ such that $\tilde A$ has the same distribution as $A$ and  $\tilde B$ has the same distribution as $B$ such that $\tilde A\leq\tilde B.$ Hence
\begin{equation*}
P (A>\alpha m)\leq P(B>\alpha m) \leq \exp(-m\phi(\alpha,\eta ))
\end{equation*}
where the second inequality follows from Chernoff's bound. Then with \eqref{eq:gmdisbn} we have
\begin{equation*}
P(\rho(\mu^{*}, \mu)>\bar{C_{\alpha}}\epsilon)\leq  \exp(-m\phi(\alpha,\eta )).
\end{equation*}
For the intrinsic case (b) we have a similar proof.
\end{proof}

 \begin{remark}\label{optmrmk}
One important aspect in constructing the estimator $\mu^*$ is the choice of the number of subsets $m$. By \eqref{eq:probgm}, a larger number of subset estimators, $m$, yields more robustness and a tighter concentration  around the true parameter.  At the same time, there must be enough data in each subset to ensure that each subset estimator behaves well and $\eta$ in \eqref{eq:etainq}  is sufficiently small. For a given  confidence level $\epsilon$,  one can determine the number of subsets to achieve $\eta$ in \eqref{eq:etainq} and the desired bound on the concentration or confidence level in \eqref{eq:probgm}.
\end{remark}

In the following, we provide examples, in both the intrinsic and extrinsic cases, of finding an $\eta$ in $\eqref{eq:etainq}$ which allows the computation of the bound in $\eqref{eq:probgm}$.

\begin{example}\label{example1}
Consider the embedding $J: \mathcal{M}\rightarrow \mathbb{R}^D.$  We have the induced measure $\tilde Q$  on the  image where $\tilde{Q}=Q\circ J^{-1}.$
Let $x_1, \ldots, x_n$ be an i.i.d. sample from a distribution $Q,$ such that we have the extrinsic mean $\mu$ for the random variable $x_1$
\begin{equation*}
\mu=J^{-1}\bigg(\mathcal{P}\Big(\int_{\mathbb{R}^D}u\tilde{Q}(du)\Big)\bigg).
\end{equation*}
Divide the sample $x_1, \ldots, x_n$ into $m$ disjoint groups $U_1, \ldots, U_m$ of size $[n/m]$
each, and define
\begin{equation*}
\begin{split}
     \tilde{\mu}_j=&\frac{1}{|U_j|}\sum_{i\in U_j}J(x_i) \quad j=1, ..., m,\\
      \mu_j&\in J^{-1}\big(\mathcal{P}(\tilde{\mu}_j)\big).
 \end{split}
\end{equation*}
One can easily conclude that
\begin{align*}
    \rho(\mu, \mu_j)&=\|J(\mu)-J(\mu_j)\|\\
   & =\|J(\mu)-\tilde{\mu}_j+\tilde{\mu}_j-J(\mu_j)\|\\
   &\leq\|J(\mu)-\tilde{\mu}_j\|+\|\tilde{\mu}_j-J(\mu_j)\|\\
   &\leq2\|J(\mu)-\tilde{\mu}_j\|.
\end{align*}
Therefore
\begin{align*}
    \mathbb{E}\rho^2(\mu, \mu_j)&\leq4\mathbb{E}\|J(\mu)-\tilde{\mu}_j\|^2\\
    &=\frac{4}{|U_j|^2}\sum_{i\in U_j}\mathbb{E}\|J(\mu)-J(x_i)\|^2\\
    &\leq\frac{4}{|U_j|^2}\sum_{i\in U_j}\mathbb{E}\rho^2(\mu, x_i)\\
    &=\frac{4}{|U_j|}\mathbb{E}\rho^2(\mu, x_1)\leq4\left[\frac{m}{n}\right]\mathbb{E}\rho^2(\mu, x_1).
\end{align*}
So by Chebyshev's inequality
\begin{equation}\label{Cheb1}
    P\big(\rho(\mu_j, \mu)\geq\epsilon\big)=P\big(\rho^2(\mu_j, \mu)\geq\epsilon^2\big) \leq\frac{1}{\epsilon^2}\mathbb{E}\rho^2(\mu_j, \mu)\leq\frac{4}{\epsilon^2}\left[\frac{m}{n}\right]\mathbb{E}\rho^2(\mu, x_1).
\end{equation}
Finally, we have the collection of independent estimators $\mu_{1}, \ldots, \mu_{m},$ such that
\begin{align*}
P(\rho(\mu_{j}, \mu)> \epsilon)\leq \eta,
\end{align*}
where $\eta=\frac{4}{\epsilon^2}\left[\frac{m}{n}\right]\mathbb{E}\rho^2(\mu, x_1).$ So by theorem \ref{thext} for any $\alpha\in(0, \cot\bar{\psi}\tan\frac{\bar{\psi}}{2})$
 \begin{align*}
P(\rho(\mu^*, \mu)>  \overbar{C}_{\alpha}\epsilon)\leq \exp(-m\phi(\alpha,\eta )),
\end{align*}
where
\begin{align*}
\mu^{*}&={\rm med}(\mu_{1}, \ldots, \mu_{m}),\\
 \overbar{C}_{\alpha}&=\frac{1-\alpha}{\sqrt{1-2\alpha}\cos\bar{\psi}-\alpha\sin\bar{\psi}},\\
\phi(\alpha, \eta)&=(1-\alpha)\log\frac{1-\alpha}{1-\eta}+\alpha\log\frac{\alpha}{\eta}.
\end{align*}

\end{example}

\begin{example}\label{example2}
Let $x_1, \ldots, x_n$ be an i.i.d. sample from a distribution $Q,$ such that we have the Fr\'echet mean $\mu$ for the random variable $x_1.$

Divide the sample $x_1, \ldots, x_n$ into $m$ disjoint groups  $U_1, \ldots, U_m$ each of size $[n/m]$, and define
\begin{equation*}
\begin{split}
     \mu_j=\arg\min_{y\in\mathcal{M}}\frac{1}{|U_j|}
     \sum_{i\in U_j}d_{g}^2(y, x_i), \quad j=1, ..., m.\\
  \end{split}
\end{equation*}
Considering the $j$th subsample corresponding to $U_j$ on the tangent space at $\mu_j$,
\begin{equation*}
    \log_{\mu_j}\mu_j=\frac{1}{|U_j|}
     \sum_{x_i\in U_j}\log_{\mu_j}x_i=0. 
\end{equation*}
Thus on the tangent space $T_{\mu_j}\mathcal{M}$, we can obtain the equality
\begin{equation*}
    d_g^2(\mu, \mu_j)=\|\log_{\mu_j}\mu\|^2=\frac{1}{|U_j|^2}\bigg\|
     \sum_{x_i\in U_j}(\log_{\mu_j}x_i-\log_{\mu_j}\mu)\bigg\|^2.
\end{equation*}
Thus,
\begin{align*}
    \mathbb{E}d_g^2(\mu, \mu_j)&=\frac{1}{|U_j|^2}\sum_{x_i\in U_j}\mathbb{E}\|\log_{\mu_j}x_i-\log_{\mu_j}\mu\|^2\\
    &\leq\frac{K^2}{|U_j|^2}\sum_{i\in U_j}\mathbb{E}d_g^2(\mu, x_i)=\frac{K^2}{|U_j|}\mathbb{E}d_g^2(\mu, x_1)\leq K^2\left[\frac{m}{n}\right]\mathbb{E}d_g^2(\mu, x_1).
\end{align*}
So by Chebyshev's inequality,
\begin{equation}\label{Cheb}
    P\big(d_g(\mu_j, \mu)\geq\epsilon\big)=P\big(d_g^2(\mu_j, \mu)\geq\epsilon^2\big) \leq\frac{1}{\epsilon^2}\mathbb{E}d_g^2(\mu_j, \mu)\leq\frac{K^2}{\epsilon^2}\left[\frac{m}{n}\right]\mathbb{E}d_g^2(\mu, x_1).
\end{equation}
Finally, we have the collection of independent estimators $\mu_{1}, \ldots, \mu_{m},$ such that
\begin{align*}
P(d_g(\mu_{j}, \mu)> \epsilon)\leq \eta,
\end{align*}
where $\eta=K^2\left[\frac{m}{n}\right]\mathbb{E}d_g^2(\mu, x_1).$ So by theorem \ref{thext} for any $\alpha\in (0, \frac{1}{2})$
 \begin{align*}
P(\rho(\mu^*, \mu)>  \overbar{C}_{\alpha}\epsilon)\leq \exp(-m\phi(\alpha,\eta )),
\end{align*}
where
\begin{align*}
\mu^{*}&={\rm med}(\mu_{1}, \ldots, \mu_{m}),\\
 \overbar{C}_{\alpha}&= K(1-\alpha)\sqrt{\frac{1}{1-2\alpha}},\\
\phi(\alpha, \eta)&=(1-\alpha)\log\frac{1-\alpha}{1-\eta}+\alpha\log\frac{\alpha}{\eta}.
\end{align*}

\end{example}

\section{Simulations and Applications} \label{sec:simapps}

In this section, through  extensive numerical examples, we show robustness and improved concentration about the population parameter of the geometric median of subset estimators in agreement with theorem \ref{thext}.
We first consider some simulated examples in estimating population means in $S^d$ and $PD(3)$. We then formulate a robust procedure for estimating explanatory directions for dimension reduction in $PD(3)$ and do a simulation study using this procedure. Finally, we apply the median-of-means method in the shape space to a hand shape data set as in \cite{FletcherVJ08}.

Numerical results from both simulated and real data analysis in this section agree with the robustness and concentration properties of the estimator.
We see in these results that
\begin{enumerate}
\item In simulations~\ref{simt1},~\ref{simt2},~\ref{simt3}, and~\ref{simt4}, and with various numbers of outliers, the average distance of the median-of-means is always an improvement over the average distances of the subset means.
\item The average distance of the median-of-means is almost always an improvement over the overall mean in the presence of outliers.
\item In the case of $PD(3)$, in Simulation~\ref{simt4}, the average distance of the median-of-means for $m=5,10, 15$ often gives an improvement over the overall median ($m=60$) in the presence of outliers. Number of groups $m=15$ seems to provide the best concentration overall. That the effect in more pronounced seems to agree with the log map in $PD(3)$ being 1-Lipschitz as in proposition~\ref{prop1ps} and with the bound given in theorem~\ref{thext} with $K=1$.
\end{enumerate}
In simulation~\ref{simt5}, the median-of-means estimator is applied in estimating both the center of operations and explanatory directions for dimension reduction. The robustness property is shown as explanatory submanifolds maintain their fit to data in terms of intrinsic sum-of-squared residuals in the presence of outliers better than the ordinary PGA procedure. All code and data used in this section can be found in ~\url{https://github.com/DrewLazar/RobustManifold}.

\subsection{Simulation Study on \texorpdfstring{$S^d$}{Sd}} \label{sec:simSd}
In this subsection, we provide examples with data simulated from the von Mises-Fisher distribution on the sphere. We consider the estimation of both intrinsic and extrinsic means in the presence of various numbers of outliers. As shown by the numerical comparisons below, the estimator obtained from the robust estimation procedure shows improved concentration over subset-based estimators and often is closer to  the true parameter of interest compared to the overall sample mean and overall sample median. We first describe algorithms used for computing various summary statistics related to our estimators in $S^d$.
\subsubsection{Computation of Sample Statistics on $S^d$} \label{sec:samplestatsinSd}
Given $\{ p_1,\ldots,p_n \} \subset S^d$  we compute sample statistics as follows:
\begin{enumerate}
\item
\textbf{Intrinsic mean.} With objective function $L_n(x)=\frac{1}{n} \sum_{i=1}^n \arccos^2(\langle x,p_i \rangle)$ and constraint function $g(x)=\langle x,x \rangle$ let
\[ \gamma_i(x) =\frac{\arccos(\langle {x},{p_i} \rangle )}{\sqrt{1-\langle {x},{p_i} \rangle^2}}. \]
Then the sample mean $\hat{\mu}$ satisfies Lagrange multiplier condition
\[
\sum_{i=1}^n \gamma_i(\hat{\mu}) p_i = \lambda \hat{\mu} \text{ with } \langle \hat{\mu}, \hat{\mu} \rangle =1  \text{ and } \lambda =  \sum_{i=1}^n \gamma_i(\hat{\mu}) \langle p_i, \hat{\mu} \rangle.
\]
As in \cite{MR2264946}, letting $\Psi (x) = \sum\nolimits_{i=1}^n \gamma_i(x) p_i$, we use the fixed-point algorithm
\begin{align*}
\mu_k & \mapsto \mu_{k+1} \\
\mu_{k+1} &  = \frac{\Psi(\mu_k)}{\norm{\Psi(\hat{\mu}_k))}}.
\end{align*}
Then $\mu_k \rightarrow \hat{\mu}$.
\item \textbf{Intrinsic median}. We use a generalization of Ostresh's modification of Weiszfeld's algorithm as introduced in \cite{FletcherVJ08}.
Let
\[ \Psi (x) = \sum_i \frac{\Log[x]{p_i}}{\arccos (\emet{x}{p_i})}  \left(\sum_i \frac{1}{\arccos (\emet{x}{p_i})}\right)^{-1}
\]
\vspace{-10pt}
\begin{align*}
m_k & \mapsto m_{k+1} \\
m_{k+1} &  = \Exp[x_k]{\Psi(m_k)}.
\end{align*}
Then $m_k \rightarrow \hat{m}$ where $\hat{m}$ is the intrinsic sample median.

\item \textbf{Extrinsic mean.} As in \cite{rabibook}, the extrinsic sample mean is the projection of the sample mean under the embedding. That is,
\[
\hat{\mu} =\mathcal{P} \left(\frac{1}{n}\sum_{i=1}^{n}J(p_i)\right)
\]
where $J$ is our embedding map. In the case of $S^d$, where $J$ is the identity map and projection is done by normalizing in $\mathbb{R}^{d+1}$, $\mu=\bar{x}/\norm{\bar{x}}$ where $\bar{x}$ is the Euclidean sample mean.

\item \textbf{Extrinsic median.} Let
\[
L_n(x)= \frac{1}{n} \sum_{i=1}^n \lvert \lvert x-p_i \rvert  \rvert \text { for }  x \in \mathbb{R}^{d+1} \text{ and } g(x)=L_n|_{S^d}.
\]
With $S^d$ as a submanifold of $\mathbb{R}^{d+1}$, for $p \in S^d$ the gradient of $g$ is the orthogonal projection of $\nabla_p L_n$ onto $T_p S^d$, that is,

    \[
    \nabla_p g = \text{proj}_{T_p S^d }(\nabla_p L_n).
    \]
We take $\nabla_p L_n$ as in Weiszfeld's algorithm \cite{MR2506479} and compute the sample geometric median $\hat m$ by gradient descent as follows:

\[
\Psi(x) = \sum_i \frac{p_i - x}{\norm{p_i - x}}\left(\sum_i \frac{1}{\norm{p_i - x}}\right)^{-1}
\]
\vspace{-5pt}
\begin{align*}
m_k & \mapsto m_{k+1} \\
m_{k+1} &  = \Exp[m_k]{\text{proj}_{T_{m_k} S^d}(\Psi(m_k))}
\end{align*}
Then $m_k \rightarrow \hat{m}$, the extrinsic sample median.
\end{enumerate}

\subsubsection{Simulations in $S^d$}
\label{sec:SimulationsSd}
We consider the von Mises-Fisher distribution on the unit sphere.
Distributions on the sphere, and the estimations of their intrinsic means have important applications in directional statistics, as in~\cite{VMDir}, and cluster analysis, as in~\cite{banerjee2005clustering}.  A von
Mises-Fisher distribution on $S^d$ has pdf
\begin{equation*}
    f_d(x; \mu, \kappa)=\frac{\kappa^{d/2-1}}{(2\pi)^{d/2}I_{d/2-1}(\kappa)}e^{\kappa\langle\mu, x\rangle},
\end{equation*}
where $I_n$ is the modified Bessel function of the first kind
\begin{equation*}
    I_n(\kappa)=\frac{2^{-n}\kappa^n}{\Gamma(n+1/2)\Gamma(1/2)}\int_0^{\pi}e^{\kappa\cos\theta}\sin^{2n}\theta d\theta.
\end{equation*}
The intrinsic mean of the distribution is $\mu$ and $\kappa$ is a concentration parameter about $\mu$ with a larger $\kappa$ giving increased concentration.  One has
\begin{equation*}
   f_d(x; \mu, \kappa)= \frac{\Gamma\big((d-1)/2\big)}{2\pi^{(d-1)/2}\int_0^{\pi}e^{\kappa\cos\theta}\sin^{d-2}\theta d\theta}e^{\kappa\langle\mu_i, x\rangle}.
\end{equation*}
Thus, sampling $x$ from the von Mises-Fisher distribution,
\begin{align} \label{von-mise}
    P\big(d_g(x, \mu)\leq\epsilon\big)&=\frac{\Gamma\big((d-1)/2\big)}{2\pi^{(d-1)/2}\int_0^{\pi}e^{\kappa\cos\theta}\sin^{d-2}\theta d\theta}\Big(\frac{2\pi^{(d-1)/2}}{\Gamma((d-1)/2)}\Big)\int_0^{\epsilon}e^{\kappa\cos\theta}\sin^{d-2}\theta d\theta \nonumber \\
     &=\frac{\int_0^{\epsilon}e^{\kappa\cos\theta}\sin^{d-2}\theta d\theta}{\int_0^{\pi}e^{\kappa\cos\theta}\sin^{d-2}\theta d\theta}.
\end{align}


 \begin{simt}{Estimating Intrinsic Mean in $S^2$:}\label{simt1}  Using \cite{MatlabVM} we sample $n=60$ data points from the von Mises-Fisher distribution on $S^2$. We take $\kappa=30$, which by \eqref{von-mise} guarantees with probability $\approx 1$ that the sample is within a hemisphere and thus the intrinsic mean and median uniquely exist.

We include $k= 0, 5, 10,\text{ and } 15$ outliers outside a symmetric $95\%$ confidence region about the mean with the confidence region computed using \eqref{von-mise}. We then apply the median-of-means technique of section 3 for $m=1, 5, 15, 30$ and $60$ groups.  Over 1000 runs, we compute
 \begin{enumerate}
  \item the average intrinsic distance $\overline{\rho(\mu^*,\mu)}$ from the true mean $\mu$ to the geometric median of subsets estimator $\mu^*$.
 \item the average intrinsic distance $\overline{\overline{\rho(\mu_i,\mu)}}$ from $\mu$ to the average of the subset means $\mu_i, \, i=1,\ldots,m$.
 \end{enumerate}

\begin{table}[h!]
\centering
\subfloat{
\begin{tabular}{{|c||c|cc|cc|cc|cc| }} \hline
k & $ \overline{\rho(\hat{\mu},\mu)}$ & $\overline{\rho(\mu^*,\mu)}$ & $\overline{\overline{\rho(\mu_i,\mu)}}$ & $\overline{\rho(\mu^*,\mu)}$ & $\overline{\overline{\rho(\mu_i,\mu)}}$ \\ \hline
0 & 0.0597 & 0.0583&0.0947 & 0.0514&0.1496 \\
5 & 0.0647  & 0.0615&0.1159 & 0.0531&0.1652  \\
10 & 0.1194  & 0.1116&0.1414 &0.1018&0.2113 \\
15 & 0.1819  & 0.1731&0.1973 & 0.1631&0.2419 \\ \hline \hline
&  sample mean (m=1) & \multicolumn{2}{c|}{m=5}  & \multicolumn{2}{c|}{m=15} \\ \hline
\end{tabular}}
\\
\subfloat{
\begin{tabular}{{|c||cc|cc| }} \hline
k & $\overline{\rho(\mu^*,\mu)}$ & $\overline{\overline{\rho(\mu_i,\mu)}}$  & $\overline{\rho(\hat{m},\mu)}$ & $\overline{\overline{\rho(\mu_i,\mu)}}$ \\ \hline
0 & 0.0455 &0.2118 & 0.0424&0.2829 \\
5 & 0.0453&0.2350 &0.0447&0.2959 \\
10 & 0.0776&0.2501 &0.0614&0.3259\\
15 & 0.1383&0.2954 &0.0925&0.3738 \\ \hline \hline
 & \multicolumn{2}{c|}{m=30} & \multicolumn{2}{c|}{sample median (m=60)} \\ \hline
\end{tabular}}
\caption{Results from Simulation 1 showing performance for various estimators of the mean under a von Mises-Fisher distribution in $S^2$, with $k$ the number of outliers and $\rho$ intrinsic distance.}
\end{table}
\textit{Note that when $m=1$, $\mu_i$ and $\mu^*$ are both the sample  Fr\'echet mean of the whole data set,  which we denote as $\hat{\mu}$}. Also, when $m=60, \mu^*$ is the sample median and $\mu_i=p_i$ for $i=1, \ldots, 60$. The same situation holds in simulations~\ref{simt2},~\ref{simt3} and~\ref{simt4}.

In Figure \ref{VM5out} we have a sample of $n=60$ from the von Mises-Fisher Distribution including 5 added outliers. We take $m=5$ subsets and we see the improved concentration about the population mean of the geometric median of the 5 subset means.
 \begin{figure}[ht!]
 \centering
  \includegraphics[scale=.3]{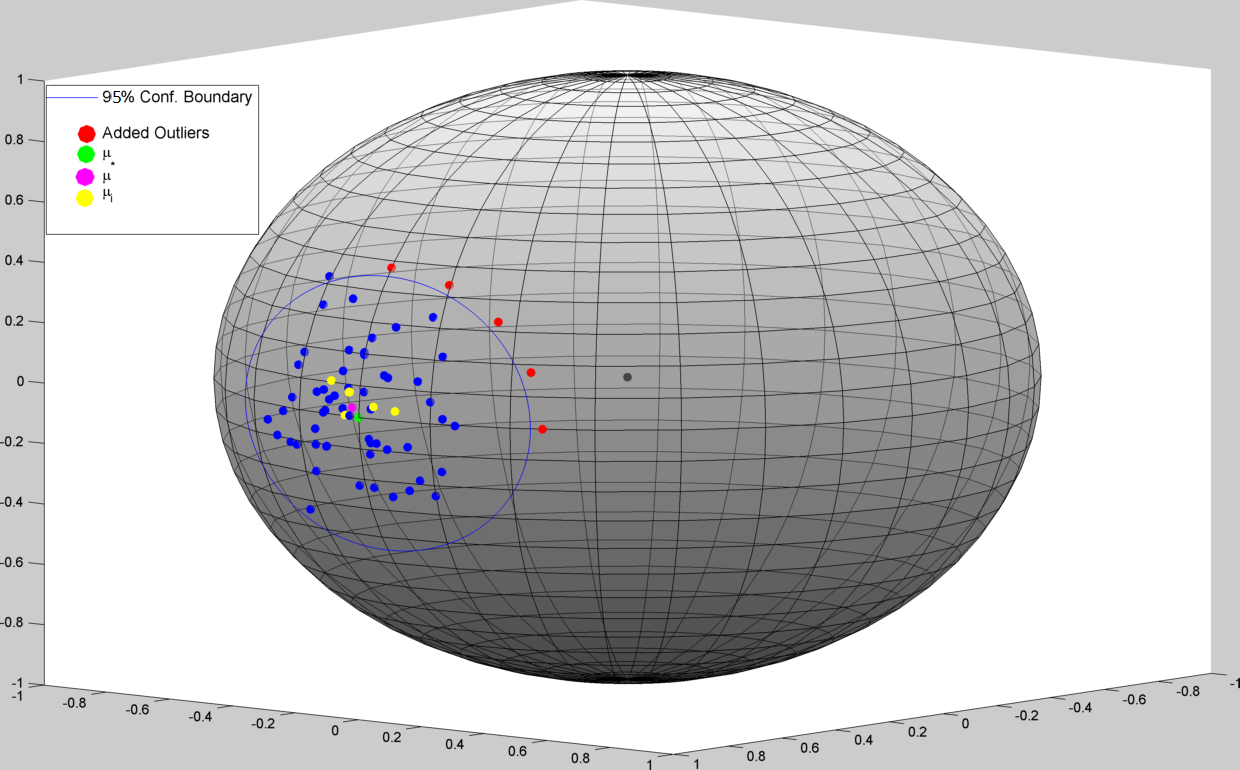}
  \caption{Von-Mises Fisher, $\kappa=30$, 5 added outliers}
  \label{VM5out}
\end{figure}
\end{simt}

\begin{simt}{Approximation of the Intrinsic Mean in $S^7$:} \label{simt2} We repeat the first part of the experiment in Simulation 1 in $S^7$ except with $n=200, \kappa=20, k=0,10,20,40$ outliers and  $m=1,10,50,100,200$ groups.
\begin{table}[h!]
\centering
\subfloat{
\begin{tabular}{{|c||c|cc|cc|cc|cc| }} \hline
k & $ \overline{\rho(\hat{\mu},\mu)}$ & $\overline{\rho(\mu^*,\mu)}$ & $\overline{\overline{\rho(\mu_i,\mu)}}$ & $\overline{\rho(\mu^*,\mu)}$ & $\overline{\overline{\rho(\mu_i,\mu)}}$ \\ \hline
0 & 0.0396  & 0.0399&0.1186 & 0.0384&0.2570 \\
10 & 0.0565 & 0.0541&0.1258 & 0.0514&0.2669  \\
20 & 0.0897  & 0.0900&0.1462 &0.0834&0.2827 \\
40 & 0.1656 & 0.1678&0.2082 & 0.1596&0.3376 \\ \hline \hline
& {Sample mean (m=1)} & \multicolumn{2}{c|}{m=10}  & \multicolumn{2}{c|}{m=50} \\ \hline
\end{tabular}}
\\
\subfloat{
\begin{tabular}{{|c||cc|cc| }} \hline
k & $\overline{\rho(\mu^*,\mu)}$ & $\overline{\overline{\rho(\mu_i,\mu)}}$  & $\overline{\rho(\hat{m},\mu)}$ & $\overline{\overline{\rho(\mu_i,\mu)}}$ \\ \hline
0 & 0.0398 &0.3590 & 0.0387&0.4896 \\
10 & 0.0469&0.3676 &0.0457&0.4978 \\
20 & 0.0760&0.3896 &0.0682&0.5301\\
40 & 0.1513&0.5176 &0.1305 &0.5987 \\ \hline \hline
 & \multicolumn{2}{c|}{m=100} & \multicolumn{2}{c|}{sample median (m=200)} \\ \hline
\end{tabular}}
\caption{Results from Simulation 2 showing performance for various estimators of the mean under a von Mises-Fisher distribution in $S^7$, with $k$ the number of outliers and $\rho$ intrinsic distance.}
\end{table}
\end{simt}
\begin{simt}{Approximation of the Extrinsic Mean in $S^2$:}\label{simt3} We repeat the experiment in Simulation 1, but with $\rho$ as the extrinsic distance and with each average taken over 1200 runs.
\begin{table}[h!]
\centering
\subfloat{
\begin{tabular}{{|c||c|cc|cc|cc|cc| }} \hline
k & $ \overline{\rho(\hat{\mu},\mu)}$ & $\overline{\rho(\mu^*,\mu)}$ & $\overline{\overline{\rho(\mu_i,\mu)}}$ & $\overline{\rho(\mu^*,\mu)}$ & $\overline{\overline{\rho(\mu_i,\mu)}}$ \\ \hline
0 & 0.0272  & 0.0330 &0.0676 & 0.0312 &0.1179 \\
5 & 0.0621  & 0.0634 &0.0943 & 0.0541&0.1512  \\
10 & 0.1231 & 0.1190 &0.1456 &0.1083&0.1952 \\
15 & 0.1771 & 0.1688 &0.1956 & 0.1632&0.2337 \\ \hline \hline
& \multicolumn{1}{c|}{Sample  mean (m=1)} & \multicolumn{2}{c|}{m=5}  & \multicolumn{2}{c|}{m=15} \\ \hline
\end{tabular}}
\\
\subfloat{
\begin{tabular}{{|c||cc|cc| }} \hline
k & $\overline{\rho(\mu^*,\mu)}$ & $\overline{\overline{\rho(\mu_i,\mu)}}$  & $\overline{\rho(\hat{m},\mu)}$ & $\overline{\overline{\rho(\mu_i,\mu)}}$ \\ \hline
0 & 0.0305 &0.1681 & 0.0312&0.2312 \\
5 & 0.0453&0.2034 &0.0411&0.2745 \\
10 & 0.0847&0.2479 &0.0612&0.3241\\
15 & 0.1453&0.2971 &0.0837&0.3728 \\ \hline \hline
 & \multicolumn{2}{c|}{m=30} & \multicolumn{2}{c|}{sample median (m=60)} \\ \hline
\end{tabular}}
\caption{Results from Simulation 3 showing performance for various estimators of the mean under a von Mises-Fisher distribution in $S^7$, with $k$ the number of outliers and $\rho$ intrinsic distance.}
\end{table}
\end{simt}
The results in Tables 1-3, showing the performance of the various estimators in Simulations 1-3 respectively, demonstrates that the median-of-means estimator almost always improves over the average of subset means and overall Fr\'echet sample mean estimators in the presence of outliers.
\subsection{Simulation study on \texorpdfstring{$PD(3)$}{PD(3)}     }
In this subsection, we consider simulated data from a generalized log-normal distribution on the space of $3\times3$ positive definite matrices, $PD(3)$. As in subsection~\ref{sec:simSd}, we consider the estimation of intrinsic means in the presence of various numbers of outliers.
There are multiple applications in which it is of interest to estimate the mean of a sample of positive definite matrices.  This includes
principal geodesic analysis (PGA), as in~\cite{FletcherJoshi}, where optimization to find explanatory directions is done in the tangent space at the sample mean.  Using our median-of-means procedure, we formulate a robust PCA procedure (RPGA).  We first describe algorithms used for
computing various summary statistics related to our estimators in $PD(3)$.

\subsubsection{Computation of Sample Statistics on $PD(3)$}
To compute the sample intrinsic mean in the following simulation, we use the damped gradient descent algorithm as in \cite{FletcherJoshi}. As shown in \cite{MR0442975}, as $PD(3)$ is of non-negative curvature, the intrinsic mean is guaranteed to exist and to be unique. To compute the sample intrinsic median, we use the generalization of Weiszfeld's algorithm given in \cite{FletcherVJ08} where the sample intrinsic median is shown to exist and to be unique. Computations of projection to subspaces and of principal geodesic directions are done using MATLAB minimization routines and user-supplied gradients as
formulated in~\cite{PGAgrad:Sommer} with the derivative of the matrix exponential map provided by~\cite[Theorem~4.5]{najfeld1995derivatives}.

\subsubsection{Robust Principal Geodesic Analysis (RPGA)}\label{sec:RPGA}

Principal Geodesic Analysis (PGA) as in~\cite{lazar2017scale} is a two-step procedure which involves 1) computing a center of the data and 2) successively finding orthogonal tangent vectors at that center so that their exponentiated span best fits the data according to intrinsic sum-of-squared residuals.

We propose a Robust PGA procedure (RPGA) which 1) uses the median-of-means estimate as the center of the data and 2) finds orthogonal directions in the tangent space using the robust median-of-means Principal Component Analysis (PCA) procedure given in~\cite{minsker2015}. Specifically, in RPGA
\begin{enumerate}
\item Divide the data into $m$ subsets $U_1,\ldots,U_m$ and for each compute an intrinsic mean $\mu_j$ as in~\eqref{eq:submeans}
and then compute $\mu^*=\text{med}(\mu_1,\ldots,\mu_m)$ as in~\eqref{eq:medianmeans}.
\item Compute $V_i =\vecc({\Log[\mu^*]{U_i}})$ where $\Log[\mu^*]{U_i}$ is the image of $U_i$ under the Riemmanian log map. As in~\cite{minsker2015}, compute sample covariance matrices $\Sigma_i$ for each $V_i$ and then compute
\[ \hat{\Sigma} = \text{med}(\Sigma_1,\ldots,\Sigma_n)
\]
where the median is taken with respect to Frobenius norm $||A||_F=\text{trace}(A^\intercal A).$ We take the eigenvectors of $\hat{\Sigma}$, $\{w_1,\ldots,w_6\}$, arranged in order by largest to smallest eigenvalue. Then our robust principal geodesic directions in the tangent space at $\mu^*$ are $\{v_1,\ldots,v_6\}$ where $v_i$ is the vector corresponding to $w_i$ by the vec operator. To form explanatory subspaces we then exponentiate the span of $\{v_1,\ldots,v_k\}$ at $\mu^*$ for $k=1,\dots,6$.

\end{enumerate}

This procedure is robust as it ensures both the located center of the data and the located explanatory directions are not unduly affected by the presence of outliers.

\subsubsection{Simulations in $PD(3)$}\label{sec:sims}
\begin{simt}{Estimating the Intrinsic Mean in $PD(3)$:} \label{simt4} We sample $n=60$ data points from a log-normal distribution where if the random variable $X$ has this distribution then $\text{vec(Log}_I(X))\sim\mathcal{N}(\mathbf{0},\kappa\mathbf{I})$ with $\kappa$ a scaling parameter. We repeat the experiment of Simulation 1 of section~\ref{sec:SimulationsSd} with each average taken over 1200 runs.
\begin{table}[h!]
\centering
\subfloat{
\begin{tabular}{{|c||c|cc|cc|}} \hline
k & $ \overline{\rho(\hat{\mu},\mu)}$ & $\overline{\rho(\mu^*,\mu)}$ & $\overline{\overline{\rho(\mu_i,\mu)}}$ & $\overline{\rho(\mu^*,\mu)}$ & $\overline{\overline{\rho(\mu_i,\mu)}}$ \\ \hline
0 & 0.2630  & 0.2781&0.5909 & 0.2753&1.0408 \\
5 & 0.2640  & 0.2512&0.5776 & 0.2683&1.0745  \\
10 & 0.3568 & 0.3179&2.7485 &0.2986&1.3158 \\
15 & 0.5292 & 0.3001&1.0433 & 0.3437&1.4246 \\ \hline \hline
& {Sample  mean (m=1) } & \multicolumn{2}{c|}{m=5}  & \multicolumn{2}{c|}{m=15} \\ \hline
\end{tabular}}
\\
\subfloat{
\begin{tabular}{{|c||cc|cc| }} \hline
k & $\overline{\rho(\mu^*,\mu)}$ & $\overline{\overline{\rho(\mu_i,\mu)}}$  & $\overline{\rho(\hat{m},\mu)}$ & $\overline{\overline{\rho(\mu_i,\mu)}}$ \\ \hline
0 & 0.2750 &1.5230 & 0.2728&2.3449 \\
5 & 0.2724&1.5930 &0.2675&2.4139 \\
10 & 0.3306&1.7607 &0.3482&2.5002\\
15 & 0.4183&1.8107 &0.5265&2.5617 \\ \hline \hline
 & \multicolumn{2}{c|}{m=30} & \multicolumn{2}{c|}{Sample median (m=60)} \\ \hline
\end{tabular}}
\caption{Results for Simulation 4 with data simulated from a log-normal distribution in $PD(3)$, $k$ the number of outliers, and $\rho$ the intrinsic distance.}
\end{table}
\end{simt}

The results are shown in Table 4.
Again, in this example, the median-of-mean estimator always improves over the average of the means and almost always over the overall sample Fr\'echet mean. The average distance from the truth of the median-of-means for $m = 5,10,15$ is an improvement over the overall median ($m = 60$) in the presence of outliers. The number of groups $m = 15$ seems to provide the best concentration overall.

\begin{simt}{Estimating Explanatory Directions in $PD(3)$ with RPGA:}\label{simt5}
We sample from a log-normal distribution,  where if the random variable $X$ has this distribution then $\text{vec(Log}_I(X))\sim\mathcal{N}(\mathbf{0},\kappa\mathbf{\Sigma})$ with $\kappa$ a scaling parameter. $\Sigma$ is diagonal with diagonal entries which vary from 1 to 20 to ensure that population PGA directions exist.

Over 200 runs, we add $0, 5, 10, 15$ outliers outside a $95\%$ confidence region in $n=60$ data points and compute PGA and RPGA explanatory directions. We then find the intrinsic mean sum of squared residuals (mSSRs) of the data without outliers relative to the estimated explanatory submanifolds. Table 5 gives the average of the mSSRs over 200 runs for submanifolds of 1, 2, and 3 dimensions for PGA and for RPGA computed with 5, 10 and 15 groups.

We see that without outliers, the PGA procedure, which sequentially optimizes a fit to the data at the intrinsic mean, produces the lowest average mSSR, regardless of the number of groups for RPGA. However, as outliers are added, the mSSR for PGA increases to a greater extent than
RPGA.  Note that RPGA with $m=1$ groups is the linear approximation of the PGA procedure given in~\cite{FletcherJoshi}
\begin{table}[h!]
\centering
\subfloat{
\begin{tabular}{{|c||c|c|c|c| }} \hline
k & \text{PGA} &  \text{RPGA} & \text{RPGA}  & \text{RPGA} \\ \hline
0 & 0.4206 & 0.4265 & 0.4259&0.4320 \\
5 & 0.4529 & 0.4465 & 0.4314&0.4342   \\
10 & 0.4541 & 0.4438 & 0.4508&0.4374 \\
15 & 0.4540 & 0.4445 & 0.4492&0.4442 \\
20 & 0.4527 & 0.4473 & 0.4507&0.4496 \\ \hline \hline
\multicolumn{2}{|c|}{$m$ groups}  & m=5  & m=10 & m=15  \\ \hline
\end{tabular}}
\hspace{10pt}
\subfloat{
\begin{tabular}{{|c||c|c|c|c| }} \hline
k & \text{PGA} &  \text{RPGA} & \text{RPGA}  & \text{RPGA} \\ \hline
0 & 0.2629 & 0.2686 & 0.2691&0.2751 \\
5 & 0.2924 & 0.2870 & 0.2803&0.2795  \\
10 & 0.2963  & 0.2838 & 0.2925&0.2791 \\
15 & 0.2994 & 0.2835 & 0.2758&0.2850 \\
20 & 0.3041 & 0.2841 & 0.2889 &0.2775 \\ \hline \hline
\multicolumn{2}{|c|}{$m$ groups}  & m=5  & m=10 & m=15  \\ \hline
\end{tabular}}
\\
\subfloat{
\begin{tabular}{{|c||c|c|c|c| }} \hline
k & \text{PGA} &  \text{RPGA} & \text{RPGA}  & \text{RPGA} \\ \hline
0 & 0.1472 & 0.1497 & 0.1533 &0.1608 \\
5 & 0.1919 & 0.1801 & 0.1600&0.1588  \\
10 & 0.2242 & 0.2102 & 0.1940&0.1743 \\
15 & 0.2208 & 0.2149 & 0.2134&0.2079 \\
20 & 0.2305 & 0.2259 & 0.2169&0.2206 \\ \hline \hline
\multicolumn{2}{|c|}{$m$ groups}  & m=5  & m=10 & m=15  \\ \hline
\end{tabular}}
\caption{Average mSSRs to explanatory submanifolds computed with $k$ outliers to data without outliers in $PD(3)$}
\end{table}
\end{simt}

\subsection{Hand Shape Data in \texorpdfstring{$\Sigma^K_2$}{SigmaK2}}
We consider the hand shape data set in \cite{CootesHand} of 18 hands with each hand in planar shape space $\Sigma^{72}_2$.  A planar shape $\Sigma_2^K$ consists of objects with $K$ landmarks in $\mathbb R^2$ modulo the Euclidean motions including rotation, scaling and translation~\cite{rabibook, kendall}.  As in \cite{FletcherVJ08}, we use ellipses as outliers with each one as
\[
\{ (a \cos (k\pi/36), b \sin(k\pi/36);  k = 0, \ldots, 71 \}
\]
 where $a, b$ are sampled from the uniform distribution on [0.5,1].
 With $k=3$ added outliers, we divide the data of size $n=21$ into $m=7$ random subsets, each of size 3. We then compute and observe the geometric median and the sample mean.

\subsubsection{Computation of Sample Statistics on $\Sigma^K_2$}
We identify $\Sigma^{72}_2$ with $S^{69}/S^1$ as in \eqref{eq:planarid}, and compute intrinsic sample means and medians using direct modifications of the algorithms in section \ref{sec:samplestatsinSd}.

In Figure~\ref{fig:sub1} (a) we show $n=21$ hands with 3 outliers. In (b) we show 7 randomly assigned subsets indicated with seven different colors, and in (c) we show the subset means of each group. In (d) we see less influence of the outliers in the geometric median, as it retains the shape of a hand similar to the original 18 hands. 
\begin{figure}[!ht]
   \centering
   \subfloat[Hand Shape Data with 3 outliers]{\includegraphics[width=.49\textwidth, trim = {12cm 1.0cm 11cm 1.0cm}, clip]{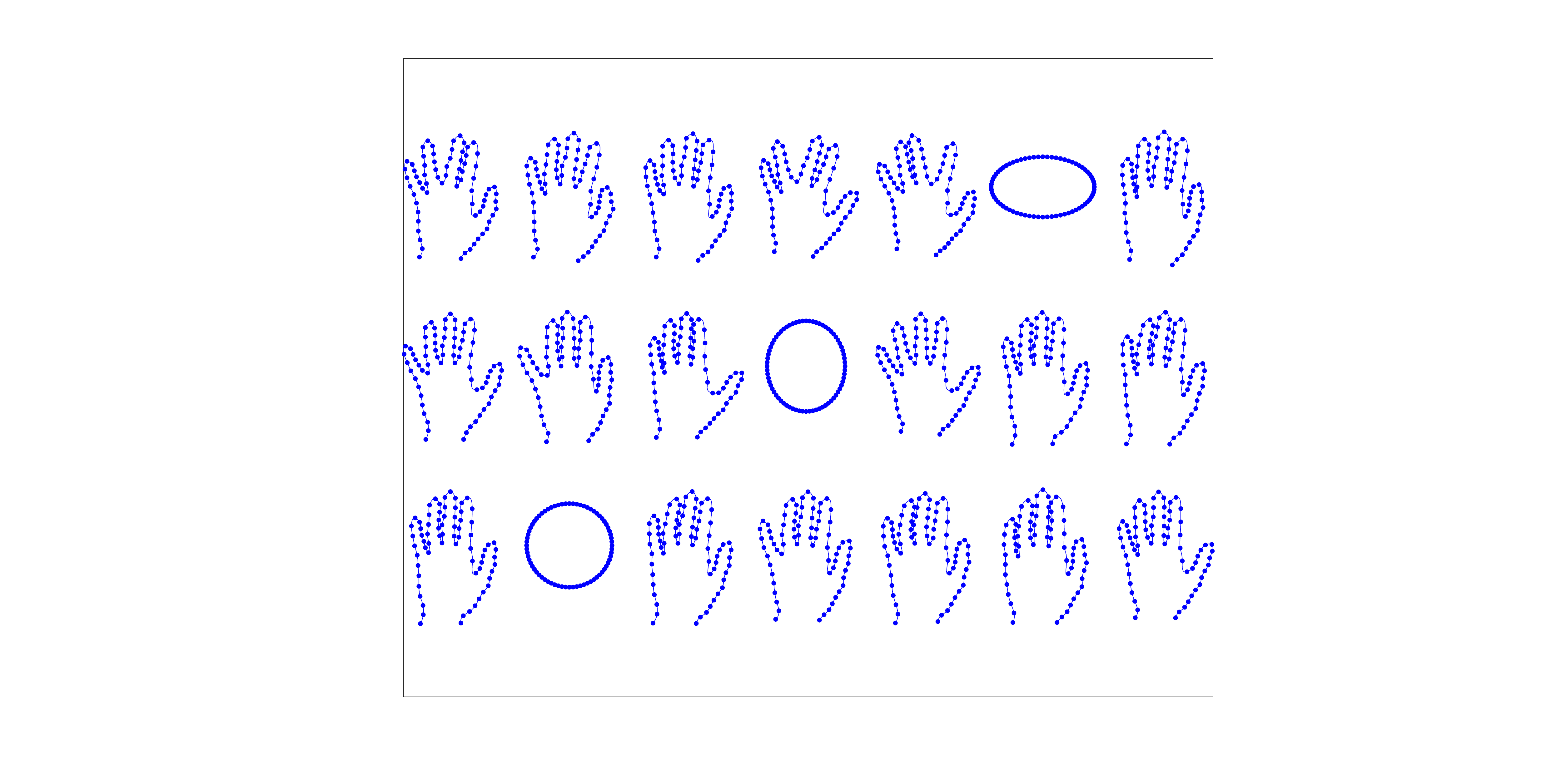}}
   \subfloat[$m=7$ subsets]{\includegraphics[width=.49\textwidth, trim = {12cm 1.0cm 11cm 1.0cm}, clip]{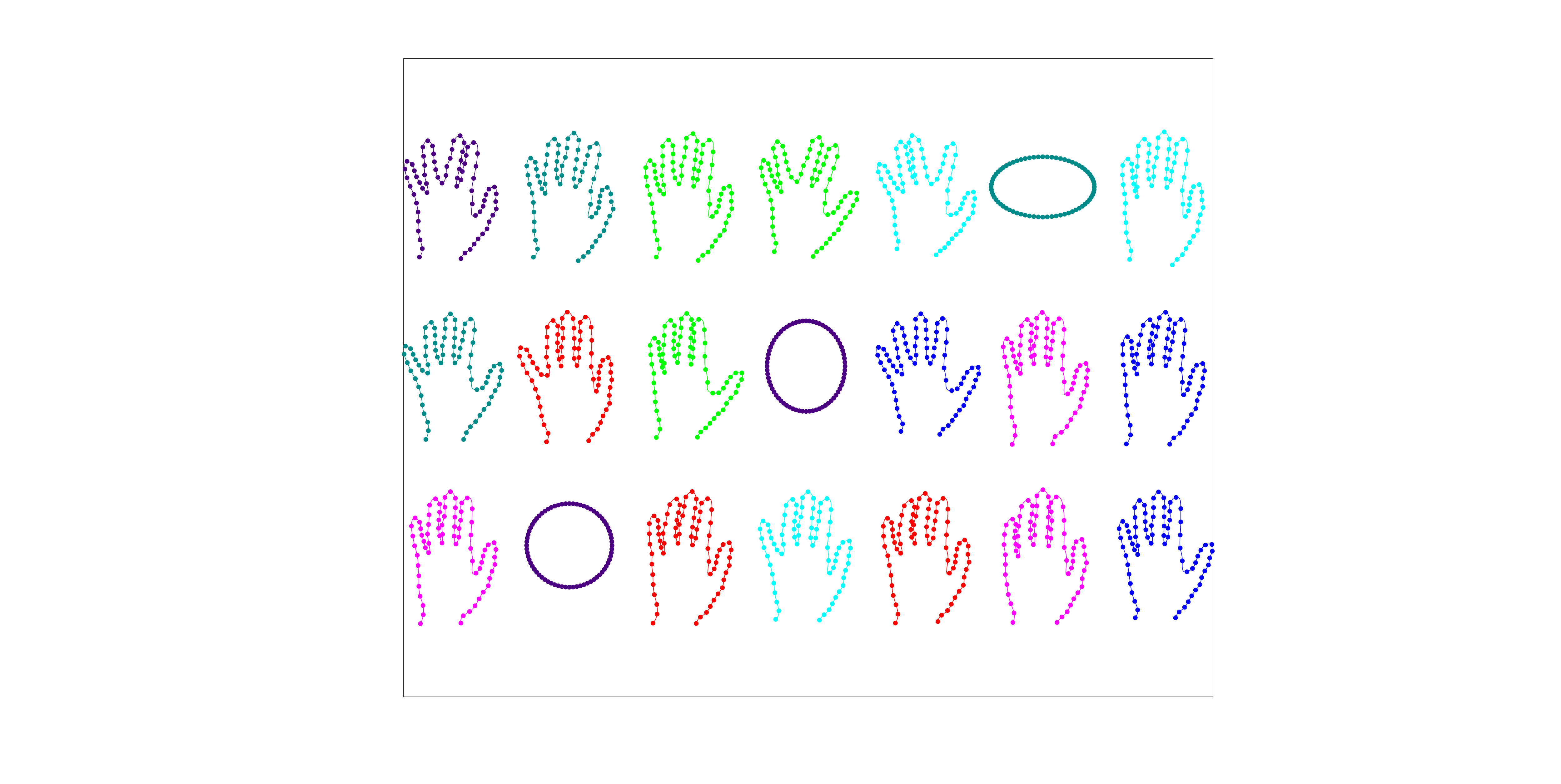}} \\
      \subfloat[Subset means, $\mu_i$]{\includegraphics[width=.49\textwidth, trim = {12cm 1.0cm 11cm 1.0cm}, clip]{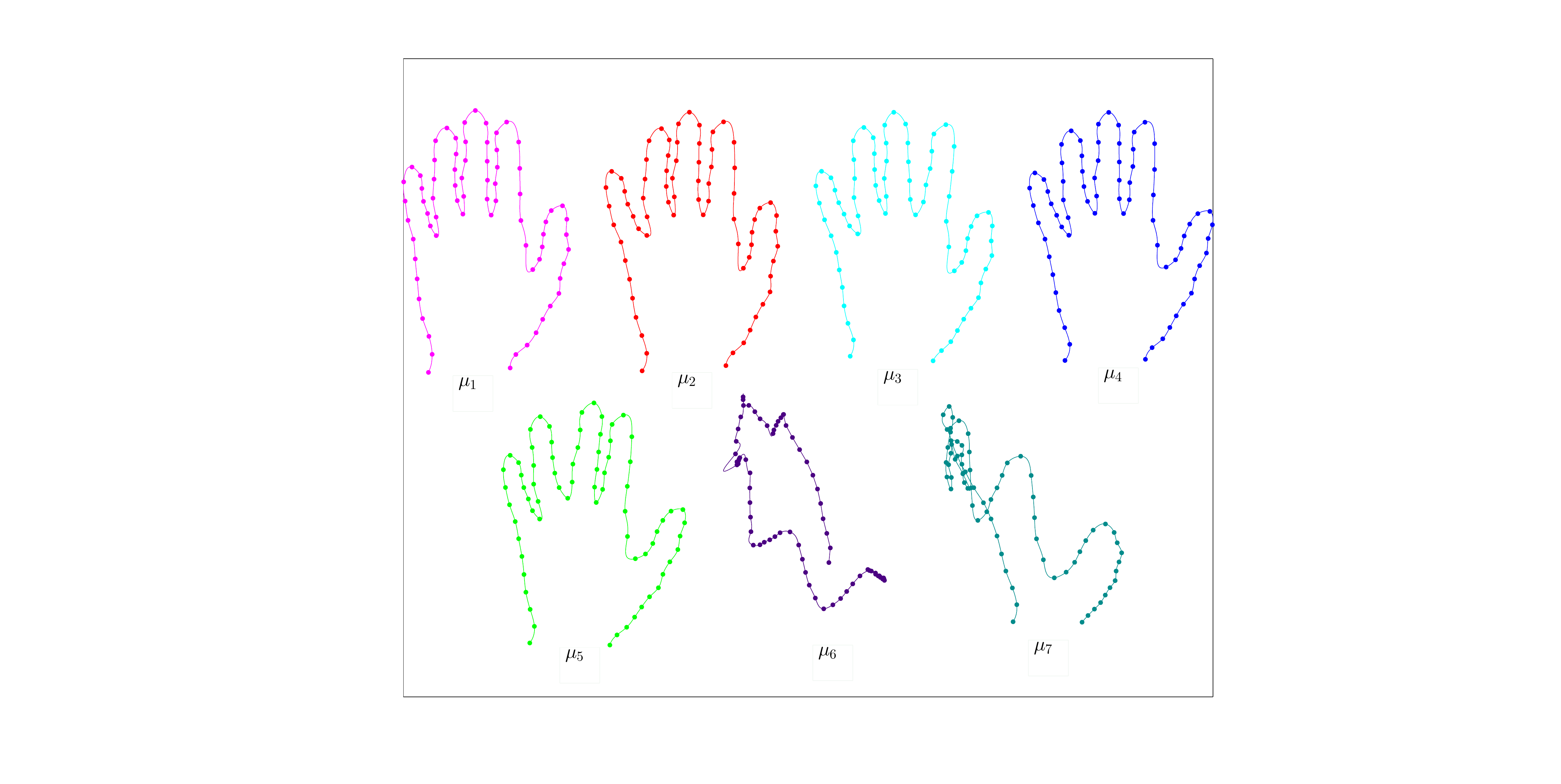}}
   \subfloat[Sample mean, $\hat{\mu}$ and geometric median, $\mu^*$ ]{\includegraphics[width=.49\textwidth, trim = {12cm 1.0cm 11cm 1.0cm}, clip]{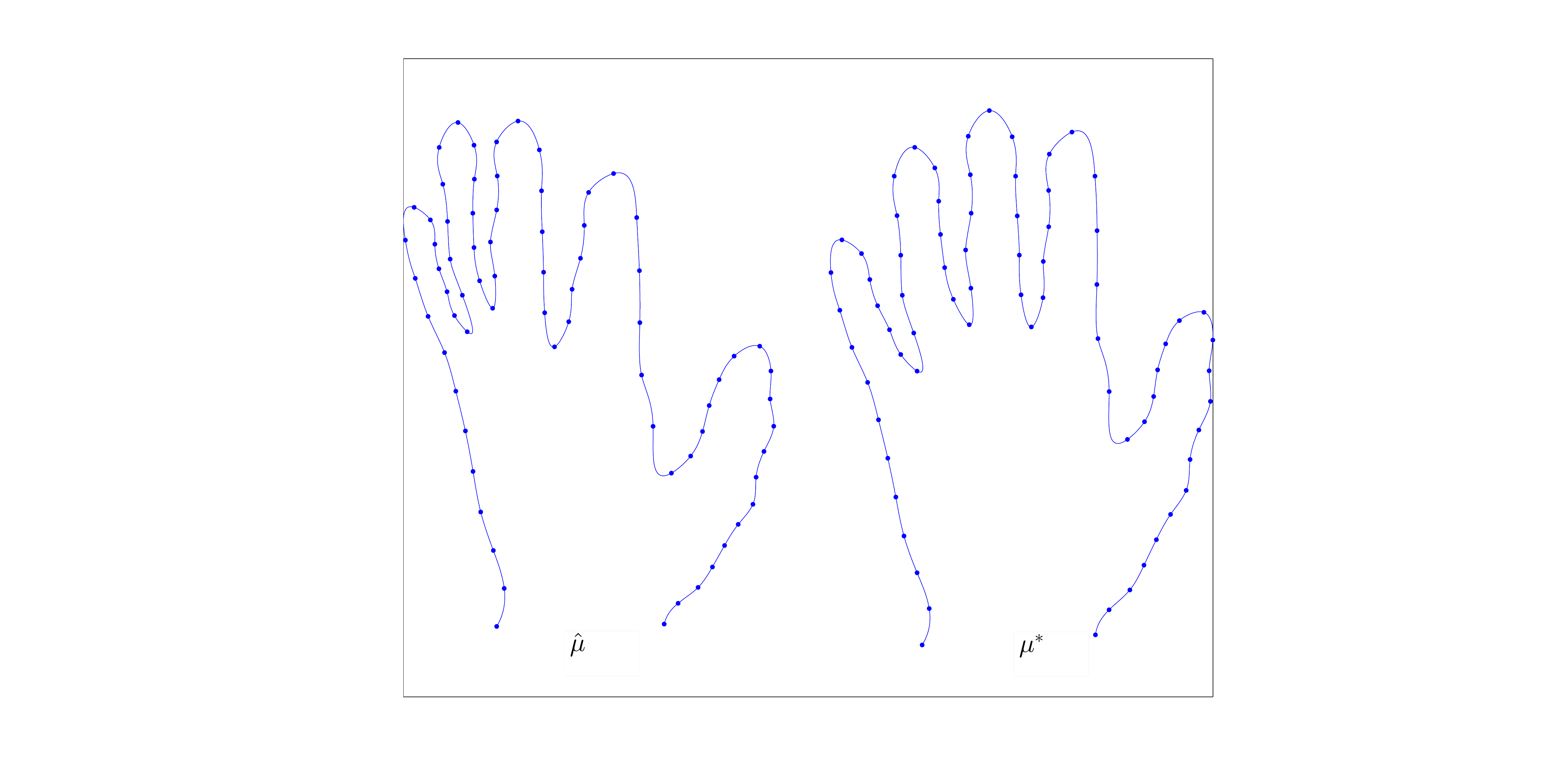}}
   \caption{Median-of-Means on Hand Shape Data}
   \label{fig:sub1}
\end{figure}


\section{Discussion}

We propose a robust and scalable procedure for general optimization problems on manifolds. Scalability is of particular importance in dealing with the difficult computational issues that arise in estimating sample statistics for manifold data or extracting low-dimensional manifold in high-dimensional data.  Along these lines, parallel computation can be implemented trivially from the subsampling procedure.

It is shown through lemma~\ref{leml1}, which provides an important property of geometric medians on manifolds, and the following theorem~\ref{thext}, that the resulting estimator yields provable robustness and tighter concentration bounds about the true parameter of interest. Numerical results from both simulated and real data analysis in Section~\ref{sec:simapps} agree with the robustness and concentration properties of the estimator.

Future research might include considering the optimal numbers and sizes of subgroups for estimation as discussed in remark~\ref{optmrmk}. In theorem~\ref{thext}, for a given $\epsilon$, more groups provide a larger $m$ but also a larger $\eta$ in the bound provided by~\ref{eq:probgm}. This is reflected in Examples~\ref{example1} and~\ref{example2} in~\eqref{Cheb1} and~\eqref{Cheb}. In finding the optimal $m$,  the number of outliers and amount of contamination in the data also must be factored in. In addition, challenging computational considerations for large data sets on manifolds and advantages in partitioning data needs to be considered. Also, the second step of the RPGA procedure in~\ref{sec:RPGA} might be done instead by partitioning the data in the manifold rather than their Riemannian logs in the tangent space at $\mu^{*}$. Computation of RPGA, as formulated in~\ref{sec:RPGA}, only requires the computation of the median-of-means $\mu^*$, and then the linear operation of computing sample covariance matrices of the Riemannian logs of data in the tangent space at $\mu^*$. Robust estimation on manifolds in other contexts such as manifold regression~\cite{aswani2011regression} might also be considered. As in the case of estimation of the mean, additional machinery and complications arise in the more general context of a manifold.

\section*{Acknowledgements}

Lizhen Lin would like to thank Dong Quan Nguyen for very useful discussions. The contribution of  LL and BS was funded by NSF grants IIS 1663870,  DMS CAREER 1654579 and a DARPA grant N66001-17-1-4041.

\bibliographystyle{plain}

\bibliography{reference-LL}

\begin{thebibliography}{10}

\bibitem{dti-ref}
Andrew~L Alexander, Jee~Eun Lee, Mariana Lazar, and Aaron~S. Field.
\newblock Diffusion tensor imaging of the brain.
\newblock {\em Neurotherapeutics}, 4(3):316--329, 2007.

\bibitem{aswani2011regression}
Anil Aswani, Peter Bickel, Claire Tomlin, et~al.
\newblock Regression on manifolds: estimation of the exterior derivative.
\newblock {\em The Annals of Statistics}, 39(1):48--81, 2011.

\bibitem{banerjee2005clustering}
Arindam Banerjee, Inderjit~S Dhillon, Joydeep Ghosh, and Suvrit Sra.
\newblock Clustering on the unit hypersphere using von mises-fisher
  distributions.
\newblock {\em Journal of Machine Learning Research}, 6(Sep):1345--1382, 2005.

\bibitem{bhatia2003exponential}
Rajendra Bhatia.
\newblock On the exponential metric increasing property.
\newblock {\em Linear Algebra and its applications}, 375:211--220, 2003.

\bibitem{rabibook}
A.~Bhattacharya and R.~Bhattacharya.
\newblock {\em Nonparametric Inference on Manifolds: With Applications to Shape
  Spaces}.
\newblock IMS Monograph \#2. Cambridge University Press, 2012.

\bibitem{linclt}
R.~{Bhattacharya} and L.~{Lin}.
\newblock {An omnibus CLT for {F}r\'echet means and nonparametric inference on
  non-Euclidean spaces}.
\newblock {\em ArXiv eprint}, 1306.5806, June 2013.

\bibitem{BOUMAL2015}
Nicolas Boumal and P.-A. Absil.
\newblock Low-rank matrix completion via preconditioned optimization on the
  grassmann manifold.
\newblock {\em Linear Algebra and its Applications}, 475:200 -- 239, 2015.

\bibitem{CootesHand}
T.F. Cootes, C.J. Taylor, D.H. Cooper, and J.~Graham.
\newblock Active shape models-their training and application.
\newblock {\em Computer Vision and Image Understanding}, 61(1):38 -- 59, 1995.

\bibitem{lowrank}
W.~{Dai}, E.~{Kerman}, and O.~{Milenkovic}.
\newblock A geometric approach to low-rank matrix completion.
\newblock {\em IEEE Transactions on Information Theory}, 58(1):237--247, 2012.

\bibitem{FletcherJoshi}
P.~Thomas Fletcher and Sarang Joshi.
\newblock Riemannian geometry for the statistical analysis of diffusion tensor
  data.
\newblock {\em Signal Process.}, 87(2):250--262, February 2007.

\bibitem{FletcherVJ08}
P.~Thomas Fletcher, Suresh Venkatasubramanian, and Sarang~C. Joshi.
\newblock Robust statistics on riemannian manifolds via the geometric median.
\newblock In {\em 2008 {IEEE} Computer Society Conference on Computer Vision
  and Pattern Recognition {(CVPR} 2008), 24-26 June 2008, Anchorage, Alaska,
  {USA}}, 2008.

\bibitem{frechet}
Maurice Fr\'echet.
\newblock L\'es \'elements al\'eatoires de nature quelconque dans un espace
  distanci\'e.
\newblock {\em Ann. Inst. H. Poincar\'e}, 10:215--310, 1948.

\bibitem{MR2264946}
Stephan Huckemann and Herbert Ziezold.
\newblock Principal component analysis for {R}iemannian manifolds, with an
  application to triangular shape spaces.
\newblock {\em Adv. in Appl. Probab.}, 38(2):299--319, 2006.

\bibitem{MatlabVM}
S.~Jung.
\newblock Random number generatrion form von mises-fisher distribution.
\newblock {\em Technical report, University of Pittsburgh}, 2010.

\bibitem{MR0442975}
H.~Karcher.
\newblock Riemannian center of mass and mollifier smoothing.
\newblock {\em Comm. Pure Appl. Math.}, 30(5):509--541, 1977.

\bibitem{kendall}
D.~G. Kendall.
\newblock Shape manifolds, {P}rocrustean metrics, and complex projective
  spaces.
\newblock {\em Bull. of the London Math. Soc.}, 16:81--121, 1984.

\bibitem{paperwitqeric}
Eric~D. Kolaczyk, Lizhen Lin, Steven Rosenberg, Jackson Walters, and Jie Xu.
\newblock Averages of unlabeled networks: Geometric characterization and
  asymptotic behavior.
\newblock {\em Ann. Statist.}, 48(1):514--538, 02 2020.

\bibitem{lazar2017scale}
Drew Lazar and Lizhen Lin.
\newblock Scale and curvature effects in principal geodesic analysis.
\newblock {\em Journal of Multivariate Analysis}, 153:64--82, 2017.

\bibitem{lecu2020}
Guillaume Lecué and Matthieu Lerasle.
\newblock Robust machine learning by median-of-means: Theory and practice.
\newblock {\em Ann. Statist.}, 48(2):906--931, 04 2020.

\bibitem{lerasle2011}
M.~{Lerasle} and R.~I. {Oliveira}.
\newblock {Robust empirical mean Estimators}.
\newblock {\em ArXiv e-prints}, December 2011.

\bibitem{sinicapaper}
Lizhen Lin, Vinayak Rao, and David~B. Dunson.
\newblock {Bayesian nonparametric inference on the Stiefel manifold}.
\newblock {\em Statistics Sinica}, 27:535--553, 2017.

\bibitem{linregression}
Lizhen Lin, Brian~St. Thomas, Hongtu Zhu, and David~B. Dunson.
\newblock Extrinsic local regression on manifold-valued data.
\newblock {\em Journal of the American Statistical Association},
  112(519):1261--1273, 2017.

\bibitem{Lohit2017LearningIR}
Suhas Lohit and Pavan~K. Turaga.
\newblock Learning invariant riemannian geometric representations using deep
  nets.
\newblock {\em 2017 IEEE International Conference on Computer Vision Workshops
  (ICCVW)}, pages 1329--1338, 2017.

\bibitem{lugosi2019}
Gábor Lugosi and Shahar Mendelson.
\newblock Regularization, sparse recovery, and median-of-means tournaments.
\newblock {\em Bernoulli}, 25(3):2075--2106, 08 2019.

\bibitem{VMDir}
K.~V. Mardia.
\newblock {\em Statistics of directional data}.
\newblock Academic Press, London-New York, 1972.
\newblock Probability and Mathematical Statistics, No. 13.

\bibitem{minsker2015}
Stanislav Minsker.
\newblock Geometric median and robust estimation in banach spaces.
\newblock {\em Bernoulli}, 21(4):2308--2335, 11 2015.

\bibitem{median-posterior}
Stanislav Minsker, Sanvesh Srivastava, Lizhen Lin, and David~B. Dunson.
\newblock Robust and scalable {B}ayes via a median of subset posterior
  measures.
\newblock {\em Journal of Machine Learning Research}, 18(124):1--40, 2017.

\bibitem{najfeld1995derivatives}
Igor Najfeld and Timothy~F Havel.
\newblock Derivatives of the matrix exponential and their computation.
\newblock {\em Advances in applied mathematics}, 16(3):321--375, 1995.

\bibitem{Nemirovski1983Problem-complex00}
A.~Nemirovski and D.~Yudin.
\newblock Problem complexity and method efficiency in optimization, 1983.
\newblock unpublished.

\bibitem{zbMATH05}
S.~P. {Novikov} and I.~A. {Taimanov}.
\newblock {\em {Modern Geometric Structures and Fields. Transl. from the
  Russian by D. Chibisov.}}, volume~71.
\newblock Providence, RI: American Mathematical Society (AMS), 2006.

\bibitem{Saparbayeva2018CommunicationEP}
Bayan Saparbayeva, Michael~Minyi Zhang, and Lizhen Lin.
\newblock Communication efficient parallel algorithms for optimization on
  manifolds.
\newblock In {\em NeurIPS}, 2018.

\bibitem{PGAgrad:Sommer}
Stefan Sommer, Fran{\c{c}}ois Lauze, and Mads Nielsen.
\newblock The differential of the exponential map, jacobi fields and exact
  principal geodesic analysis.
\newblock {\em CoRR, abs/1008.1902}, 2010.

\bibitem{MR2506479}
E.~Weiszfeld.
\newblock On the point for which the sum of the distances to {$n$} given points
  is minimum.
\newblock {\em Ann. Oper. Res.}, 167:7--41, 2009.
\newblock Translated from the French original [Tohoku Math. J. {{\bf{4}}3}
  (1937), 355--386] and annotated by Frank Plastria.

\end{thebibliography}

\end{document}